\newif{\ifarxiv}
\arxivtrue%

\documentclass[orivec,envcountsect,envcountsame]{llncs}

\usepackage{silence}
\usepackage{amsmath}
\usepackage{amssymb}

\usepackage{arydshln}
\usepackage{amsthm}
\usepackage{stmaryrd}
\usepackage{marvosym}
\usepackage{mathpartir}
\usepackage{upgreek}
\usepackage{microtype}

\usepackage{tikz}
\usetikzlibrary{positioning}

\usepackage{aliascnt}
\usepackage[bookmarks,unicode,colorlinks=true]{hyperref}

\usepackage{paralist}
\usepackage{xspace}

\usepackage{thm-restate}
\usepackage{cleveref}

\SetSymbolFont{stmry}{bold}{U}{stmry}{m}{n}

\newcommand{\alphabet}{\Sigma}

\newcommand{\BA}{\ensuremath{\mathsf{BA}}\xspace}
\newcommand{\KA}{\ensuremath{\mathsf{KA}}\xspace}

\newcommand{\CKAO}{\ensuremath{\mathsf{CKAO}}\xspace}

\newcommand{\termst}[1]{\ensuremath{{\mathcal{T}_{\scriptscriptstyle#1}}}}
\newcommand{\termsba}{\termst{\BA}}
\newcommand{\termska}{\termst{\KA}}
\newcommand{\termscka}{\termst{}}
\newcommand{\termsckao}{\termst{\CKAO}}

\newcommand{\equivt}[1]{\equiv_{\scriptscriptstyle#1}}
\newcommand{\equivba}{\equivt{\BA}}
\newcommand{\equivka}{\equivt{\KA}}
\newcommand{\equivcka}{\equivt{}}

\newcommand{\leqqt}[1]{\leqq_{\scriptscriptstyle#1}}
\newcommand{\leqqba}{\leqqt{\BA}}

\newcommand{\leqqcka}{\leqq}

\newcommand{\semt}[2]{{\left\llbracket#2\right\rrbracket}_{\scriptscriptstyle#1}}

\newcommand{\semka}{\semt{\KA}}
\newcommand{\semcka}{\semt{}}

\newcommand{\angl}[1]{\left\langle#1\right\rangle}
\newcommand{\pipe}{\;\;|\;\;}

\newcommand{\lp}[1]{\mathbf{#1}}
\newcommand{\ltr}[1]{\mathtt{#1}}
\newcommand{\pl}[1]{\mathcal{#1}}

\newcommand{\naturals}{\mathbb{N}}

\newcommand{\N}{\ensuremath{\mathsf{N}}\xspace}

\newcommand{\Pom}{\ensuremath{\mathsf{Pom}}}
\newcommand{\SP}{\ensuremath{\mathsf{SP}}}

\newcommand{\PC}{\mathsf{PC}}
\newcommand{\PCsp}{\mathsf{PC}^\mathsf{sp}}
\newcommand{\PCseq}{\mathsf{PC}^\mathsf{seq}}

\newcommand{\subsp}{\sqsubseteq^\mathsf{sp}}

\newcommand{\seq}{\mathsf{seq}}

\newcommand{\hexch}{\mathsf{exch}}
\newcommand{\hbool}{\mathsf{bool}}
\newcommand{\hcontr}{\mathsf{contr}}
\newcommand{\hglu}{\mathsf{glue}}
\newcommand{\hobs}{\mathsf{obs}}
\newcommand{\hgrp}{\mathsf{group}}

\newcommand{\At}{\mathsf{At}}

\newcommand\closure[2][H]{{#2}{\downarrow^{#1}}}
\newcommand\closurep[2][H]{\left(#2\right){\downarrow^{#1}}}
\newcommand\closurepsmall[2][H]{(#2){\downarrow^{#1}}}
\newcommand\seqclosure[2][H]{{#2}{\downarrow^{#1}_\seq}}

\spnewtheorem{fact}{Fact}{\bfseries}{\itshape} 

\usepackage{todonotes}

\makeatletter
   \def\@citecolor{blue}%
   \def\@urlcolor{blue}%
   \def\@linkcolor{blue}%

\def\orcidID#1{\smash{\href{http://orcid.org/#1}{\protect\raisebox{-1.25pt}{\protect\includegraphics{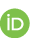}}}}}
\makeatother

\title{Concurrent Kleene Algebra with Observations:\texorpdfstring{\\}{} from Hypotheses to Completeness}

\author{
    Tobias Kapp\'e~\orcidID{0000-0002-6068-880X} (\Letter) \and 
    Paul Brunet \orcidID{0000-0002-9762-6872} \and 
    Alexandra Silva \orcidID{0000-0001-5014-9784} \and\texorpdfstring{\\}{} 
    Jana Wagemaker \orcidID{0000-0002-8616-3905} \and 
    Fabio Zanasi \orcidID{0000-0001-6457-1345} 
}
\institute{University College London, London, United Kingdom; \email{tkappe@cs.ucl.ac.uk}}

\ifarxiv%
\else%
\fi%

\ifarxiv%

\else%

\fi%

\begin{document}

\maketitle

\begin{abstract} Concurrent Kleene Algebra (CKA) extends basic Kleene algebra with a parallel composition operator, which enables reasoning about concurrent programs.
However, CKA fundamentally misses \emph{tests}, which are needed to model standard programming constructs such as conditionals and $\mathsf{while}$-loops.
It turns out that integrating tests in CKA is subtle, due to their interaction with parallelism.
In this paper we provide a solution in the form of Concurrent Kleene Algebra with Observations (CKAO).
Our main contribution is a completeness theorem for CKAO\@.
Our result resorts on a more general study of CKA ``with hypotheses'', of which CKAO turns out to be an instance: this analysis is of independent interest, as it can be applied to extensions of CKA other than CKAO\@.
 \end{abstract}

\section{Introduction}%
\label{section:introduction}
\emph{Kleene algebra with tests} (KAT) is a (co)algebraic framework~\cite{kozen-1996,kozen-2017} that allows one to study properties of imperative programs with conditional branching, i.e. $\mathsf{if}$-statements and $\mathsf{while}$-loops. 
KAT is build on Kleene algebra (KA)~\cite{conway-1971,kozen-1994}, the algebra of regular languages.
Both KA and KAT enjoy a rich meta-theory, which makes them a suitable foundation for reasoning about program verification.
In particular, it is well-known that the equational theories of KA and KAT characterise rational languages~\cite{salomaa-1966,krob-1990,kozen-1994} and guarded rational languages~\cite{kozen-1996} respectively.
Efficient procedures for deciding equivalence have been studied in recent years, also in view of recent applications to network verification~\cite{bonchi-pous-2013,foster-kozen-etal-2015,smolka-foster-etal-2020}.

Concurrency is a known source of bugs and hence challenges for verification.
Hoare, Struth, and collaborators~\cite{hoare-moeller-struth-wehrman-2009},  have proposed an extension of KA, \emph{Concurrent Kleene Algebra} (CKA), as an algebraic foundation for concurrent programming.
CKA enriches the basic language of KA with a parallel composition operator $\cdot \parallel \cdot$.
Analogously to  KA, CKA also has a semantic characterisation for which the equational theory is complete, in terms of rational languages of \emph{pomsets} (words with a partial order on letters)~\cite{laurence-struth-2014,laurence-struth-2017-arxiv,kappe-brunet-silva-zanasi-2018}.

The development of CKA raises a natural question, namely how tests, which were  essential in KAT for the study of sequential programs, can be integrated into CKA\@.
At first glance, the obvious answer may appear to be to merge KAT with CKA, yielding Concurrent Kleene Algebra with Tests (CKAT) --- as attempted in~\cite{jipsen-moshier-2016}.
However, as it turns out, integrating tests into CKA is quite subtle and this naive combination does not adequately capture the behaviour of concurrent programs.
In particular, using the CKAT framework of~\cite{jipsen-moshier-2016} one can prove that for any test $b$ and CKAT program $e$:\[
    0
        \quad \leq_{\scriptscriptstyle \mathsf{KAT}} \quad b \cdot e \cdot \overline{b}
        \quad \leq_{\scriptscriptstyle \mathsf{CKA}} \quad  e \parallel (b \cdot \overline{b})
        \quad \equiv_{\scriptscriptstyle \mathsf{KAT}} \quad  e \parallel 0
        \quad  \equiv_{\scriptscriptstyle \mathsf{CKA}} \quad  0
\]
thus $b \cdot e \cdot \overline{b} \equiv_{\scriptscriptstyle \mathsf{CKAT}} 0$, meaning no program $e$ can change the outcome of any test $b$.
Or equivalently, and undesirably, that any test is an invariant of any program!

The core issue is the identification in KAT of sequential composition $\cdot$ and Boolean conjunction $\land$.
In the concurrent setting this is not sound as the values of variables --- and hence tests --- can be changed between the two tests.

In order to fix this issue, we have presented \emph{Kleene Algebra with Observations} (KAO) in previous work~\cite{kappe-brunet-rot-silva-wagemaker-zanasi-2019}.
Algebraically, KAO differs from KAT in that conjunction of tests $b \wedge b'$ and their sequential composition $b \cdot b'$ are distinct operations.
In particular, $b \wedge b'$ expresses a single test executed \emph{atomically}, whereas $b \cdot b'$ describes two distinct executions, occurring one after the other.
As mentioned above, this distinction is crucial   when moving from the sequential setting of KA to the concurrent setting of CKA, as actions from another thread that happen to be scheduled after $b$ but before $b'$ may as well change the outcome of $b'$.

This newly developed extension of KA enables a novel attempt to enrich CKA with the ability to reason about programs that also have the traditional conditionals: in this paper, we present Concurrent Kleene Algebra with Observations (CKAO) and show that it overcomes the problems present in CKAT\@.

The traditional plan for developing a variant of (C)KA is to define a separate syntax, semantics, and set of axioms, before establishing a formal correspondence with the base syntax, semantics and axioms of (C)KA proper, and arguing that this correspondence allows one to conclude soundness and completeness of the axioms w.r.t.\ the semantics, as well as decidability of equivalence in the semantics. 
Instead of such a tailor-made proof, however, we take a more general approach by first proposing CKA with hypotheses (CKAH) as a formalism for studying extensions of CKA, akin to how Kleene algebra with hypotheses~\cite{cohen-1994,kozen-2002,kozen-mamouras-2014,doumane-kuperberg-pous-pradic-2019} can be used to extend Kleene algebra.
We then apply CKAH to study CKAO, but the meta-theory developed can also be applied to extensions other than CKAO\@.

Using the CKAH formalism, we instantiate CKAO as CKAH with a particular set of hypotheses, and we immediately obtain a syntax and semantics; we can then use the meta-theory of CKAH to argue completeness and decidability in a modular proof, which composes results about CKA~\cite{kappe-brunet-silva-zanasi-2018} and KAO~\cite{kappe-brunet-rot-silva-wagemaker-zanasi-2019}.

The technical roadmap of the paper and its contributions are as follows.
\begin{itemize}
    \item
    We introduce Concurrent Kleene Algebra with Hypotheses (CKAH), a formalism for studying extensions of CKA;\@ this is a concurrent extension of Kleene Algebra with Hypotheses (\autoref{section:ckahypotheses}).
    We show how CKAH is sound with respect to rational pomset languages closed under an operation arising from the set of hypotheses.
    We propose techniques to argue completeness of the extended set of axioms with respect to the sound model as well as decidability of equivalence, capturing methods commonly used in literature to argue completeness and decidability for extensions of (concurrent) KA\@.

    \item
    We prove that CKAO can be presented as an instance of CKAH, for a certain set of hypotheses (\autoref{section:instantiationtockao}).
    This gives us a sound model of CKAO `for free'.
    We then prove that the axioms of CKAO are also complete for this model, and that equivalence is decidable, using the techniques developed previously.
\end{itemize}

\noindent
We conclude this introduction by giving an example of how hypotheses can be added to CKA to include the meaning of primitive actions.
Suppose we were designing a DSL for recipes, specifically, the steps necessary, and their order.
A recipe to prepare cookies might contain the actions $\mathsf{mix}$ (mixing the ingredients), $\mathsf{preheat}$ (pre-heating the oven), $\mathsf{chill}$ (chilling the dough) and $\mathsf{bake}$ (baking the cookies).
Using these actions, a recipe like ``mix the ingredients until combined; chill the dough while pre-heating the oven; bake cookies in the oven'' may be encoded as $\mathsf{mix}^* \cdot (\mathsf{chill} \parallel \mathsf{preheat}) \cdot \mathsf{bake}$.
Now, imagine that we have only one oven, meaning that we cannot bake two batches of cookies concurrently.
We might encode this restriction on concurrent behaviour by forcing the equation
\[
    (e \cdot \mathsf{bake} \cdot f) \parallel (g \cdot \mathsf{bake} \cdot h)
        = (e \cdot \mathsf{bake} \parallel g) \cdot (f \parallel \mathsf{bake} \cdot h) +
        (e \parallel g \cdot \mathsf{bake}) \cdot (\mathsf{bake} \cdot f \parallel h)
\]
As a consequence of this hypothesis, one could then derive properties such as
\[
    \mathsf{bake} \parallel (\mathsf{bake} \cdot \mathsf{mix}) = \mathsf{bake} \cdot \mathsf{bake} \cdot \mathsf{mix} + \mathsf{bake} \cdot \mathsf{mix} \cdot \mathsf{bake}
\]

In a nutshell, this paper provides an algebraic framework --- CKAH --- together with techniques for soundness and completeness results.
The framework is flexible in that different instantiations of the hypotheses generate very different algebraic systems.
We provide one instantiation --- CKAO --- that enables analysis of programs with both concurrency primitives and Boolean assertions.
This is the first sound and complete algebraic theory to reason about such programs.

For the sake of brevity, some proofs appear in
\ifarxiv%
\hyperref[appendix:proofs]{Appendix~\ref{appendix:proofs}}.
\else%
the extended version~\cite{arxiv}.
\fi%

\section{Preliminaries}%
\label{section:preliminaries}
We recall basic definitions on pomset languages, used in the semantics of CKA, which generalise languages to allow letters in words to be partially ordered.
We fix a (possibly infinite) alphabet $\Sigma$.
When defining sets parametrised by $\Sigma$, say $\mathsf S(\Sigma)$, if $\Sigma$ is clear from the context we use $\mathsf S$ to refer to $\mathsf S(\Sigma)$.

\subsubsection{Posets and Pomsets}

Pomsets~\cite{gischer-1988,grabowski-1981} are labelled posets, up to isomorphism. 
\begin{definition}[Labellet poset]%
\label{definition:lp}
A \emph{labelled poset} over $\Sigma$ is a tuple $\lp{u} = \angl{S, \leq, \lambda}$, where $S$ is a finite set (the \emph{carrier} of $\lp{u}$), $\leq_\lp{u}$ is a partial order on $S$ (the \emph{order} of $\lp{u}$), and $\lambda: S \to \alphabet$ is a function (the \emph{labelling} of $\lp{u}$).
\end{definition}

We will denote labelled posets by bold lower-case letters $\lp{u}$, $\lp{v}$, etc.
We write $S_\lp{u}$ for the carrier of $\lp{u}$, $\leq_\lp{u}$ for the order of $\lp{u}$, and $\lambda_\lp{u}$ for the labelling of $\lp{u}$.
We assume that any labelled poset has a carrier that is a subset of some countably infinite set, say $\naturals$; this allows us to speak about the \emph{set of labelled posets} over $\Sigma$.
The precise contents of the carrier, however, are not important --- what matters to us is the labels of the points, and the ordering between them.

\begin{definition}[Poset isomorphism, pomset]%
\label{definition:lp-isomorphism-pomset}
Let $\lp{u}, \lp{v}$ be labelled posets over $\Sigma$.
We say $\lp{u}$ is \emph{isomorphic} to $\lp{v}$, denoted $\lp{u} \cong \lp{v}$, if there exists a bijection $h: S_\lp{u} \to S_\lp{v}$ that preserves labels, and preserves and reflects ordering.
More precisely, we require that $\lambda_\lp{v} \circ h = \lambda_\lp{u}$, and $s \leq_\lp{u} s'$ if and only if $h(s) \leq_\lp{v} h(s')$.

A \emph{pomset} over $\Sigma$ is an isomorphism class of labelled posets over $\Sigma$, i.e., the class $[\lp{v}] = \{ \lp{u} : \lp{u} \cong \lp{v} \}$ for some labelled poset $\lp{v}$.
\end{definition}

We write $\Pom(\Sigma)$ for the set of pomsets over $\Sigma$, and $1$ for the empty pomset.
As long as we have countably many pomsets in scope, the above allows us to assume w.l.o.g.\ that those pomsets are represented by labelled posets with pairwise disjoint carriers; we tacitly make this assumption throughout this paper.

Pomsets can be concatenated, creating a new pomset that contains all events of the operands, with the same label, but which orders all events of the left operand before those of the right one. We can also compose pomsets in parallel, where events of the operands are juxtaposed without any ordering between them.

\begin{definition}[Pomset composition]%
\label{definition:pomset-composition}
Let $U = [\lp{u}]$ and $V = [\lp{v}]$ be pomsets over $\Sigma$.
We write $U \parallel V$ for the \emph{parallel composition} of $U$ and $V$, which is the pomset over $\Sigma$ represented by the labelled poset $\lp{u} \parallel \lp{v}$, where
\begin{mathpar}
S_{\lp{u} \parallel \lp{v}} = S_\lp{u} \cup S_\lp{v}
\and
\leq_{\lp{u} \parallel \lp{v}} = {\leq_\lp{u}} \cup {\leq_\lp{v}}
\and
\lambda_{\lp{u} \parallel \lp{v}}(x) =
\begin{cases}
  \lambda_\lp{u}(x)&x \in S_\lp{u} \\
  \lambda_\lp{v}(x)&x \in S_\lp{v}
\end{cases}
\end{mathpar}

Similarly, we write $U \cdot V$ for the \emph{sequential composition} of $U$ and $V$, that is, the pomset represented by the labelled poset $\lp{u} \cdot \lp{v}$, where
\begin{mathpar}
S_{\lp{u} \cdot \lp{v}} = S_{\lp{u} \parallel \lp{v}}
\and
{\leq_{\lp{u} \cdot \lp{v}}} = {\leq_\lp{u}} \cup {\leq_\lp{v}} \cup (S_\lp{u} \times S_\lp{v})
\and
\lambda_{\lp{u} \cdot \lp{v}} = \lambda_{\lp{u} \parallel \lp{v}}
\end{mathpar}
\end{definition}

Just like words are built up from the empty word and letters using concatenation, we can build a particular set of pomsets using only sequential and parallel composition; this will be the primary type of pomset that we will use.

\begin{definition}[Series-parallel]%
\label{definition:pomset-sp}
The set of \emph{series-parallel pomsets} (\emph{sp-pomsets}) over $\Sigma$, denoted $\SP(\alphabet)$, is the smallest set s.t.\ $1 \in \SP(\alphabet)$, $\ltr{a} \in \SP(\alphabet)$ for every $\ltr{a} \in \alphabet$, and it is closed under parallel and sequential composition.
\end{definition}

The following characterisation of $\SP$ is very useful in proofs.

\begin{theorem}[Gischer~\cite{gischer-1988}]
Let $U = [\lp{u}] \in \Pom$.
Then $U \in \SP$ if and only if $U$ is \emph{\N-free}, which is to say that if there exist no distinct $s_0, s_1, s_2, s_3 \in S_\lp{u}$ such that $s_0 \leq_\lp{u} s_1$ and $s_2 \leq_\lp{u} s_3$ and $s_0 \leq_\lp{u} s_3$, with no other relation between them.
\end{theorem}

One way of comparing pomsets is to see whether they have the same events and labels, except that one is ``more sequential'' in the sense that more events are ordered.
This is captured by the notion of \emph{subsumption}~\cite{gischer-1988}, defined as follows.

\begin{definition}[Subsumption]%
\label{definition:subsumption}
Let $U = [\lp{u}]$ and $V = [\lp{v}]$.
We say $U$ \emph{is subsumed by} $V$, written $U \sqsubseteq V$, if there exists a label- and order-preserving bijection $h: S_\lp{v} \to S_\lp{u}$.
That is, $\lambda_\lp{u} \circ h = \lambda_\lp{v}$ and if $s \leq_\lp{v} s'$, then $h(s) \leq_\lp{u} h(s')$.
\end{definition}

Subsumption between sp-pomsets can be characterised as follows~\cite{gischer-1988}.

\begin{lemma}%
\label{lemma:subsumption-sp}
Let $\subsp$ be $\sqsubseteq$ restricted to $\SP$.
Then $\subsp$ is the smallest precongruence (preorder monotone w.r.t.\ the operators) such that for all $U, V, W, X \in \SP$:
\[(U \parallel V) \cdot (W \parallel X) \subsp (U \cdot W) \parallel (V \cdot X)\]
\end{lemma}

\subsubsection{CKA:\@ syntax and semantics.} 

CKA terms are generated by the grammar
\[
    e, f\in\termscka(\Sigma) ::= 0 \pipe 1 \pipe \ltr{a} \in \alphabet \pipe e + f \pipe e \cdot f \pipe e \parallel f \pipe e^*
\]

Semantics of CKA is given in terms of \emph{pomset languages}, that is subsets of $\SP$, which we simply denote by $2^\SP$. Formally, the function $\semcka{-}: \termscka \to 2^\SP$ assigning languages to CKA terms is defined as follows:
\begin{align*}
\semcka{0} &= \emptyset
    & \semcka{1} &= \{ 1 \}&
 \semcka{e + f} &= \semcka{e} \cup \semcka{f}
    & \semcka{e \cdot f} &= \semcka{e} \cdot \semcka{f}    \\
    \semcka{e^*} &= \semcka{e}^*    &
\semcka{\ltr{a}} &= \{ \ltr{a} \}
    & \semcka{e \parallel f} &= \semcka{e} \parallel \semcka{f}
\end{align*}
Here, we use the pointwise lifting of sequential and parallel composition from pomsets to pomset languages, i.e., when $\pl{U}, \pl{V} \subseteq \SP(\alphabet)$, we define
\begin{mathpar}
\pl{U} \cdot \pl{V} = \{ U \cdot V : U \in \pl{U}, V \in \pl{V} \}
\and
\pl{U} \parallel \pl{V} = \{ U \parallel V : U \in \pl{U}, V \in \pl{V} \}
\end{mathpar}
Furthermore, the Kleene star of a pomset language $\pl{U}$ is defined as
$\pl{U}^* = \bigcup_{n \in \naturals} \pl{U}^n$, where
$\pl{U}^0 = \{ 1 \}$ and $
\pl{U}^{n+1} = \pl{U}^n \cdot \pl{U}$.

Equivalence of CKA terms can be axiomatised in the style of Kleene algebra. The relation $\equivcka$ is the smallest congruence on $\termscka$ (with respect to all operators) such that for all $e, f, g \in \termscka$:
\begin{mathpar}
e + 0 \equivcka e
\and
e + e \equivcka e
\and
e + f \equivcka f + e
\and
e + (f + g) \equivcka (f + g) + h
\\
e \cdot (f \cdot g) \equivcka (e \cdot f) \cdot g
\and
e \cdot (f + g) \equivcka e \cdot f + e \cdot h
\and
(e + f) \cdot g \equivcka e \cdot g + f \cdot g
\\
e \cdot 1 \equivcka e \equivcka 1 \cdot e
\and
e \cdot 0 \equivcka 0 \equivcka 0 \cdot e
\and
e \parallel f \equivcka f \parallel e
\and
e \parallel 1 \equivcka e
\and
e \parallel 0 \equivcka 0
\\
 e \parallel (f \parallel g) \equivcka (e \parallel f) \parallel g
\and
e \parallel (f + g) \equivcka e \parallel f + e \parallel g
\and
1 + e \cdot e^* \equivcka e^* \equivcka 1 + e^* \cdot e
\\
e + f \cdot g \leqqcka g \implies f^* \cdot e \leqqcka g
\and
e + f \cdot g \leqqcka f \implies e \cdot g^* \leqqcka f
\end{mathpar}
in which $e \leqqcka f$ is the natural order $e + f \equivcka f$. %
The final (conditional) axioms are referred to as the \emph{least fixpoint axioms}.

Laurence and Struth~\cite{laurence-struth-2014} proved this axiomatisation to be sound and complete. A decision procedure was proposed in~\cite{brunet-pous-struth-2017}.
\begin{theorem}[Soundness, completeness, decidability]%
\label{theorem:bka-completeness-decidability}
Let $e, f \in \termscka$. We have:
    $e \equivcka f$ if and only if $\semcka{e} = \semcka{f}$, and
    it is decidable whether $\semcka{e} = \semcka{f}$.
\end{theorem}
Readers familiar with CKA will notice that the algebra defined here is not in fact CKA as defined in~\cite{hoare-moeller-struth-wehrman-2009}.
Indeed the signature axiom of CKA, the exchange law, has been omitted.
However, as we show in \autoref{sec:fact-exch-law}, the standard definition of CKA, as well as its completeness proof~\cite{kappe-brunet-silva-zanasi-2018}, may be recovered using hypotheses.

\section{Pomset contexts}
The linear one-dimensional structure of words makes it straightforward to define occurrences of subwords: if one wants to state that a word $w$ appears in another word $v$, one can simply say that $v = xwy$ for some $x$ and $y$.
Due to the two-dimensional nature of pomsets, it is not straightforward to define when a pomset occurs inside another pomset, because the pomset could appear below a parallel, which is nested in a sequential, which is in a parallel, etc.
In what follows we define {\em pomset contexts}, that will enable us to talk about pomset factorisations in a similar fashion as we do for words, and prove some useful properties for these.

\begin{definition}
Let $*$ be a symbol not occurring in $\Sigma$.
A \emph{pomset context} is a pomset over $\Sigma \cup \{ * \}$ with exactly \emph{one} node labelled by $*$.
More precisely, $C$ is a pomset context if $C = [\lp{c}]$ with exactly one $s_* \in S_\lp{c}$ with $\lambda_\lp{c}(s_*) = *$.
\end{definition}
Intuitively, $*$ is a placeholder or gap where another pomset can be inserted. We write $\PC(\Sigma)$ for the set of pomset contexts over $\Sigma$, and $\PCsp(\Sigma)$ for the series-parallel pomset contexts over $\Sigma$.

Given a $C \in \PC$ and $U \in \Pom$, we can ``plug'' $U$ into the gap left in $C$ to obtain the pomset $C[U] \in \Pom$.
More precisely, let $U = [\lp{u}]$ and $C = [\lp{c}]$ with $\lp{u}$ disjoint from $\lp{c}$.
We write $C[U]$ for the pomset represented by $\lp{c}[\lp{u}]$, where $S_{\lp{c}[\lp{u}]} = S_\lp{u} \cup S_\lp{c} - \{ * \}$ and $\lambda_{\lp{c}[\lp{u}]}(s)$ is given by $\lambda_\lp{c}(s)$ if $s \in S_\lp{c} - \{ * \}$, and $\lambda_\lp{u}(s)$ when $s \in S_\lp{u}$; lastly, $\leq_{\lp{c}[\lp{u}]}$ is the smallest relation on $S_{\lp{c}[\lp{u}]}$ satisfying
\begin{mathpar}
\inferrule{%
    s \leq_\lp{u} s'
}{%
    s \leq_{\lp{c}[\lp{u}]} s'
}
\and
\inferrule{%
    s \leq_\lp{c} s'
}{%
    s \leq_\lp{\lp{c}[\lp{u}]} s'
}
\and
\inferrule{%
    s_* \leq_\lp{c} s \\
    s' \in S_\lp{u}
}{%
    s' \leq_{\lp{c}[\lp{u}]} s
}
\and
\inferrule{%
    s' \in S_\lp{u} \\
    s \leq_\lp{c} s_*
}{%
    s \leq_{\lp{c}[\lp{u}]} s'
}
\end{mathpar}
It follows easily that $\leq_{\lp{c}[\lp{u}]}$ is a partial order.
We may also apply contexts to languages: if $L\subseteq\Pom$ and $C\in\PC$, the language $C[L]$ is defined as $\left\{C[U]:U\in L\right\}$.

We now prove some properties of contexts that will be useful later in our technical development. First, we note that pomset contexts respect subsumption.

\begin{restatable}{lemma}{pomsetcontextmonotone}%
\label{lemma:pomset-context-monotone}
Let $C, D \in \PC$, $U \in \Pom$.
If $C \sqsubseteq D$, then $C[U] \sqsubseteq D[U]$.
\end{restatable}
Series-parallel pomset contexts can be given an inductive characterisation.

\begin{restatable}{lemma}{seriesparallelcontextsinductive}%
\label{lemma:series-parallel-contexts-inductive}
$\PCsp$~is the smallest pomset language $L$ satisfying
\begin{mathpar}
\inferrule{~}{%
    * \in L
}
\and
\inferrule{%
    U \in \SP \\
    C \in L
}{%
    U \cdot C \in L
}
\and
\inferrule{%
    C \in L \\
    V \in \SP
}{%
    C \cdot V \in L
}
\and
\inferrule{%
    U \in \SP \\
    C \in L
}{%
    U \parallel C \in L
}
\end{mathpar}
\end{restatable}

We will identify \emph{totally ordered pomsets} with words, i.e., $\Sigma^* \subseteq \SP$.
If the pomset $U$ inserted in a context $C$ is a non-empty word, and the resulting pomset is a parallel pomset, then we can infer how to factorise $C$.
\begin{restatable}{lemma}{contextprimelocus}%
\label{lemma:context-prime-locus}
Let $C \in \PCsp$ be a pomset context, let $V, W \in \Pom$, and let $U \in \Sigma^*$ be non-empty.
If $C[U] = V \parallel W$, then there exists a $C' \in \PCsp$ such that either $C = C' \parallel W$ and $C'[U] = V$, or $C = V \parallel C'$ and $C'[U] = W$.
\end{restatable}

Application of series-parallel contexts preserves series-parallel pomsets.
\begin{restatable}{lemma}{contextpreserveseriesparallel}%
\label{lemma:context-preserve-series-parallel}
Let $C \in \PCsp$.
If $U \in \SP$, then $C[U] \in \SP$ as well.
\end{restatable}

If we plug the empty pomset into a context, then any subsumed pomset can be obtained by plugging the empty pomset into a subsumed context.
If the subsumed pomset is series-parallel, then so is the subsumed context.

\begin{restatable}{lemma}{contexterasure}%
\label{lemma:context-erasure}
Let $C \in \PC$ and $V \in \Pom$ with $V \sqsubseteq C[1]$.
We can construct $C' \in \PC$ such that $C' \sqsubseteq C$ and $C'[1] = V$.
Moreover, if $V \in \SP$, then $C' \in \PCsp$.
\end{restatable}

An analogue to the previous lemma can be obtained if instead of the empty pomset one inserts a single letter pomset $a$.
\begin{restatable}{lemma}{contextsubsumption}%
\label{lemma:context-subsumption}
Let $C \in \PC$, $V \in \Pom$ and $\ltr{a} \in \Sigma$ with $V \sqsubseteq C[\ltr{a}]$.
We can construct $C' \in \PC$ s.t.\ $C' \sqsubseteq C$ and $C'[\ltr{a}] = V$.
Moreover, if $V \in \SP$, then $C' \in \PCsp$.
\end{restatable}

\section{Concurrent Kleene Algebra with Hypotheses}%
\label{section:ckahypotheses}

Kleene algebra has basic axioms about how program composition operators should work in general, and hence does not make any assumptions about how these operators work on specific programs.
When reasoning about equivalence in a programming language, however, it makes sense to embed domain-specific truths about the operators into the axioms.
For instance, if a programming language includes assignments to variables, %
then subsequent assignments to the same variable could be merged into one, giving rise to an equation such as
\begin{equation}%
\label{eq:double-assign}
x \leftarrow m \leq x \leftarrow n \cdot x \leftarrow m,
\end{equation}
which says that the behaviour of first assigning $n$, then $m$ to $x$ (on the right) includes the behaviour of simply assigning $m$ to $x$ directly (on the left).

Kleene algebra with hypotheses (KAH)~\cite{cohen-1994,kozen-2002,kozen-mamouras-2014,doumane-kuperberg-pous-pradic-2019} enables the addition of extra axioms, called \emph{hypotheses}, to the axioms of KA\@.
The appeal of KAH is that it allows a wide range of such hypotheses about programs to be added to the equational theory, while retaining the theoretical boilerplate of KA\@.
In particular, it turns out that we can derive a sound model for any set of hypotheses, using the language model that is sound for KA proper~\cite{doumane-kuperberg-pous-pradic-2019}.
Moreover, the completeness and decidability results that hold for KA can be leveraged to obtain completeness and decidability results for some specific types of hypotheses~\cite{cohen-1994,kozen-mamouras-2014,doumane-kuperberg-pous-pradic-2019}; in general, equivalence under other hypotheses may turn out to be undecidable~\cite{kozen-2002}.

In this section, we propose a generalisation of so-called Kleene algebra with hypotheses to a concurrent setting, showing how one can obtain a sound (pomset language) model for any set of hypotheses.
We then discuss a number of techniques that allow one to prove completeness and decidability of the resulting system for a large set of hypotheses, by relying on  analogous results about CKA\@.

\begin{definition}
A \emph{hypothesis} is an inequation $e \leq f$ where $e, f \in \termscka$.
When $H$ is a set of hypotheses, we write $\equivcka^H$ for the smallest congruence on $\termscka$ generated by the hypotheses in $H$ as well as the axioms and implications that build $\equivcka$.
More concretely, whenever $e \leq f \in H$, also $e \leqqcka^H f$.%
\end{definition}

A hypothesis that declares two programs to be equivalent, such as in~\eqref{eq:double-assign}, can be encoded by including both $e \leq f$ and $f \leq e$ in $H$.

\begin{example}
Suppose the set of primitive actions $\Sigma$ includes the increments of the form $\mathtt{incr}\, x$, as well as a statement $\mathtt{print}$, which writes the complete state of the machine (including variables) on the standard output.
Since we would like to depict the state consistently, the state should not change while the output is rendered; hence, $\mathtt{print}$ cannot be executed concurrently with any other action.
Instead, when a program containing $\mathtt{print}$ is scheduled to run in parallel with an assignment, it must be interleaved such that the assignment runs either entirely before or after $\mathtt{print}$.
To encode this, we can include in $H$ the hypotheses
\[
    \mathtt{incr}\, x \parallel \mathtt{print}
        = \mathtt{incr}\, x \cdot \mathtt{print} + \mathtt{print} \cdot \mathtt{incr}\, x
\]
for all variables $x$.
This allows us to prove, for instance, that
\[
    \mathtt{print} \cdot \mathtt{incr}\, x \cdot \mathtt{incr}\, x \cdot \mathtt{print}
        \leqqcka^H {(\mathtt{incr}\, x \parallel \mathtt{print})}^*
\]
That is, if we run some number of increments and $\mathtt{print}$ statements in parallel, it is possible that $x$ is incremented twice between print statements.
\end{example}

To obtain a model of CKAH, it is not enough to use $\semcka{-}$, as some programs equated by the hypotheses might have different semantics.
To get around this, we adapt the method from~\cite{doumane-kuperberg-pous-pradic-2019}: take $\semcka{-}$ as a base semantics, and adapt the resulting language using hypotheses, such that the pomsets that could be obtained by rearranging the term using the hypotheses are also present in the language:

\begin{definition}
Let $L \subseteq \Pom$.
We define the $H$-closure of $L$, written $\closure L$, as the smallest language containing $L$ such that for all $e \leq f \in H$ and $C \in \PCsp$, if $C[\semcka f]\subseteq \closure L$, then $C[\semcka e]\subseteq\closure L$.
Formally, $\closure L$ may be described as the smallest language satisfying the following inference rules:
\begin{mathpar}
    \inferrule{~}{%
        L\subseteq \closure L
    }
    \and
    \inferrule{%
      e\leq f \in H\\
      C\in\PCsp\\
      C[\semcka f]\subseteq \closure L
    }{%
      C[\semcka e]\subseteq \closure L
    }
\end{mathpar}
\end{definition}

\begin{example}
Continuing with the hypotheses $H$ and actions $\Sigma$ used in the previous examples, note that if $L = \semcka{\mathtt{incr}\, x \parallel \mathtt{print}}$, then we have that
\[
    \mathtt{incr}\, x \parallel \mathtt{print}\in\closure L
\]
Choose $C = *$; because $\mathtt{incr}\, x \cdot \mathtt{print} + \mathtt{print} \cdot \mathtt{incr}\, x \leq \mathtt{incr}\, x \parallel \mathtt{print} \in H$ and for all $U \in \semcka{\mathtt{incr}\, x \parallel \mathtt{print}}$ we have $C[U]\in L \subseteq \closure L$, also
\[
    C[\mathtt{incr}\, x \cdot \mathtt{print}] = \mathtt{incr}\, x \cdot \mathtt{print}\in\closure L
\]
and therefore $\mathtt{incr}\, x \cdot \mathtt{print} \in \closure L$.
\end{example}

We observe the following useful properties about the interaction between closure and other operators on pomset languages.

\begin{restatable}{lemma}{compositionversusclosure}%
\label{lemma:composition-vs-closure}
Let $L, K \subseteq \Pom$ and $C \in \PCsp$.
The following hold.

\smallskip
\begin{minipage}{0.40\textwidth}
\begin{enumerate}[\hspace{-.56cm}1.]
    \item\label{property:hypothesis-closure}
    $L \subseteq {\closure K}$ iff ${\closure L} \subseteq {\closure K}$.

    \item\label{property:hypothesis-monotone}
    If $L \subseteq K$, then ${\closure L} \subseteq {\closure K}$.

    \item\label{property:hypothesis-union}
    ${\closurep {L \cup K}} = {\closurep {{\closure L} \cup {\closure K}}}$

    \item\label{property:hypothesis-concat}
    ${\closurep {L \cdot K}} = {\closurep {{\closure L} \cdot {\closure K}}}$
\end{enumerate}
\end{minipage}
\begin{minipage}{0.54\textwidth}
\begin{enumerate}[\hspace{-.05cm}1.]
 \setcounter{enumi}{4}
    \item\label{property:hypothesis-parallel}
    ${\closurep {L \parallel K}} = {\closurep {{\closure L} \parallel {\closure K}}}$

    \item\label{property:hypothesis-star}
    ${\closurep {L^*}} = \closure {({\left(\closure L\right)}^*)}$

    \item\label{property:hypothesis-context}
    If $\closure{L} \subseteq \closure{K}$, then $\closure{C[L]} \subseteq \closure{C[K]}$.

    \item\label{property:hypothesis-sp}
    If $L \subseteq \SP$, then $\closure{L} \subseteq \SP$.
\end{enumerate}
\end{minipage}
\end{restatable}

\begin{remark}
Property~\eqref{property:hypothesis-closure} states that $\closure{-}$ is a closure operator.
However, it is not in general a Kuratowski closure operator~\cite{kuratowski-1922}, since it fails to commute with union.
For instance, let $\ltr{a}, \ltr{b}, \ltr{c} \in \Sigma$ and $H = \{ \ltr{a} \leq \ltr{b} + \ltr{c} \}$; then $\closure {\{ \ltr{b} \}} \cup \closure {\{ \ltr{c} \}} = \{ \ltr{b}, \ltr{c} \}$, while $\ltr{a} \in \closurep {\{ \ltr{b} \} \cup \{ \ltr{c} \}}$.
\end{remark}

Using~\autoref{lemma:composition-vs-closure}, we can show that, if we combine the semantics from $\semcka{-}$ with $H$-closure, we obtain a sound semantics for CKA with hypotheses $H$.

\begin{restatable}[Soundness]{lemma}{soundness}%
\label{lemma:soundness}
If $e \equivcka^H f$, then $\closure {\semcka{e}} = \closure {\semcka{f}}$.
\end{restatable}

The converse of the above, where semantic equivalence is sufficient to establish axiomatic equivalence, is called \emph{completeness}.
Similarly, we may also be interested in \emph{deciding} whether $\closure {\semcka{e}}$ and $\closure {\semcka{f}}$ coincide.

\begin{definition}
Let $e, f \in \termscka$.
\begin{enumerate}[(i)]
    \item
    If $\closure {\semcka{e}} = \closure {\semcka{f}}$ implies $e \equivcka^H f$, then $H$ is called \emph{complete}.

    \item
    If $\closure {\semcka{e}} = \closure {\semcka{f}}$ is decidable, then $H$ is said to be \emph{decidable}.
\end{enumerate}
\end{definition}

Note that, in the special case where $H = \emptyset$, we know that $H$ is complete and decidable by \autoref{theorem:bka-completeness-decidability}.
One method to find out whether $H$ is complete or decidable is to reduce the problem to this special case.
More concretely, suppose we know $\closure{\semcka{e}} = \closure{\semcka{f}}$, and want to establish that $e \equivcka^H f$.
If we could find a set of hypotheses $H'$ that is complete, and we could map $e$ and $f$ to terms $r(e)$ and $r(f)$ such that $\closure[H']{\semcka{r(e)}} = \closure[H']{\semcka{r(f)}}$, then we would have $r(e) \equivcka^{H'} r(f)$.
If we could then ``lift'' that equivalence to prove $e \equivcka^H f$, we are done.
Similarly, if we would know that $\closure[H']{\semcka{r(e)}} = \closure[H']{\semcka{r(f)}}$ is equivalent to $\closure[H]{\semcka{e}} = \closure[H]{\semcka{f}}$, we could decide the latter. To formalise this intuition, we first need the following.

\begin{definition}
  We say that $H$ \emph{implies} $H'$ if we can use the hypotheses in $H$ to prove those of $H'$, i.e., if for every hypothesis $e \leq f \in H'$ it holds that $e \leqqcka^H f$.
\end{definition}

Implication relates to equivalence and closure as follows.

\begin{restatable}{lemma}{implicationlemma}%
\label{lemma:implication}
Let $H$ and $H'$ be sets of hypotheses such that $H$ implies $H'$.
\begin{enumerate}[(i)]
    \item\label{property:implication-vs-equivalence}
    If $e, f \in \termscka$ with $e \equivcka^{H'} f$, then $e \equivcka^H f$.

    \item\label{property:implication-vs-closure}
    If $L \subseteq \Pom$, then $\closure[H']{L} \subseteq \closure{L}$.

    \item\label{property:implication-vs-double-closure}
    If $L \subseteq \Pom$, then $\closurepsmall{\closure[H']{L}} = \closure{L}$.
\end{enumerate}
\end{restatable}

If $H$ implies $H'$ and vice versa, then $H$ is complete (resp.\ decidable) precisely when $H'$ is.
In general, however, this is not very helpful; we need something more asymmetrical, in order to get from a complicated set of hypotheses $H$ to a simpler set of hypotheses $H'$, where completeness or decidability might be easier to prove.
Ideally, we would like to reduce to $H' = \emptyset$, which is complete and decidable.

One idea to formalise this idea of a reduction is as follows.

\begin{definition}%
\label{definition:reduction}
Let $H$ and $H'$ be sets of hypotheses such that $H$ implies $H'$.
A map $r: \termscka \to \termscka$ is a \emph{reduction} from $H$ to $H'$ when both of the following are true:
\begin{enumerate}[(i)]
    \item
    for $e \in \termscka$, it holds that $e \equivcka^H r(e)$, and

    \item
    for $e, f \in \termscka$, if $\closure[H]{\semcka{e}} = \closure[H]{\semcka{f}}$, then $\closure[H']{\semcka{r(e)}} = \closure[H']{\semcka{r(f)}}$.
\end{enumerate}
We call $H$ \emph{reducible} to $H'$ if there exists a reduction from $H$ to $H'$.
\end{definition}

It is straightforward to show that reductions do indeed carry over completeness and decidability results, in the following sense.

\begin{restatable}{lemma}{carryingover}%
\label{lemma:carry-over}
Suppose $H$ is reducible to $H'$.
The following hold:
\begin{enumerate}[(i)]
    \item
    If $H'$ is complete, then so is $H$.

    \item
    If $H'$ is decidable, then so is $H$.
\end{enumerate}
\end{restatable}

\begin{example}%
\label{example:simple-reduction}
Let $\Sigma = \{ \ltr{a}, \ltr{b} \}$.
Let $H = \{ \ltr{a} \leq \ltr{b} \}$.
We can define for $e \in \termscka$ the term $r(e) \in \termscka$, which is $e$ but with every occurrence of $\ltr{b}$ replaced by $\ltr{a} + \ltr{b}$.
For instance, $r(\ltr{a} \cdot \ltr{b}^* \parallel \ltr{c}) = \ltr{a} \cdot {(\ltr{a} + \ltr{b})}^* \parallel \ltr{c}$.
An inductive argument on the structure of $e$ shows that $r$ reduces $H$ to $\emptyset$, and hence $H$ is complete and decidable.
\end{example}

It is not very hard to show that reductions can be chained, as follows.
\begin{restatable}{lemma}{chains}%
\label{lemma:chains}
If $H$ reduces to $H'$, which reduces to $H''$, then $H$ reduces to $H''$.
\end{restatable}

Another way of reducing $H$ is to find two sets of hypotheses $H_0$ and $H_1$, and reduce each of those to another set of hypotheses $H'$~\cite{doumane-kuperberg-pous-pradic-2019}.
The idea is that a proof of $e \equivcka^H f$ can be split up in a phase where we find $e', f' \in \termscka$ such that $e \equivcka^{H_0} e'$ and $f \equivcka^{H_0} f'$, after which we find $e'', f'' \in \termscka$ with $e' \equivcka^{H_1} e''$ and $f' \equivcka^{H_1} f''$.
Finally, we establish that $e'' \equivcka^{H'} f''$, before lifting those equivalences to $H$, concluding
\[
    e \equivcka^H e' \equivcka^H e'' \equivcka^H f'' \equivcka^H f' \equivcka^H f
\]
One way of achieving this is as follows.

\begin{definition}
We say that $H$ \emph{factorises} into $H_0$ and $H_1$ if $H$ implies both $H_0$ and $H_1$, and for all $L \subseteq \SP$ we have that $\closure[H]{L} = \closure[H_1]{(\closure[H_0]{L})}$.
\end{definition}

In order to use factorisation to compose simpler reductions into more complicated ones, we need a slightly stronger notion of reduction, as follows.

\begin{definition}
We say that $r$ is a \emph{strong reduction} from $H$ to $H'$ if it is a reduction such that for $e \in \termscka$, it holds that $\closure[H]{\semcka{e}} = \closure[H']{\semcka{r(e)}}$.
\end{definition}

Note that this additional condition essentially strengthens the second condition in \autoref{definition:reduction}.
Factorisation then lets us compose strong reductions.

\begin{restatable}{lemma}{factorisation}%
\label{lemma:factorisation}
Suppose $H$ factorises into $H_0$ and $H_1$, and both $H_0$ and $H_1$ strongly reduce to $H'$.
Then $H$ strongly reduces to $H'$.
\end{restatable}

The remainder of this section is devoted to developing techniques that can be used to design reductions, based on the properties of the sets of hypotheses under consideration.
Using the lemmas we have established so far, these techniques may then be leveraged to obtain completeness and decidability results.

\subsection{Reification}
It can happen that the hypotheses in $H$ impose an algebraic structure on the letters in $\Sigma$; for instance, as we will see later on, the letters in $H$ could be propositional terms, whose equivalence is mediated by the axioms of Boolean algebra.
In order to peel away this layer of axioms and reduce to a smaller $H'$, we can try to reduce to terms over a smaller alphabet, making the algebraic structure on the letters irrelevant to equivalence.
In a sense, performing this kind of reduction is like showing that the equivalences between letters from the hypotheses can already be guaranteed by replacing them with the right terms.

\begin{example}%
\label{example:reify-group}
Let $\Sigma$ be the set of group terms over a (finite) alphabet $\Lambda$, that is, $\Sigma$ consists of the terms generate by the grammar
\[
    g, h ::= u \pipe \ltr{a} \in \Lambda \pipe g \circ h \pipe \overline{g}
\]
Furthermore, let $\equiv_G$ be the smallest congruence generated by the group axioms, i.e., for all $g, h, i \in \Lambda$ it holds that
\begin{mathpar}
g \circ (h \circ i) \equiv_G (g \circ h) \circ i
\and
g \circ u \equiv_G g \equiv_G u \circ g
\and
\overline{g} \circ g \equiv_G u \equiv_G g \circ \overline{g}
\end{mathpar}
Lastly, let $\hgrp = \{ g \leq h : g \equiv_G h \}$.
We can then define a reduction from $\hgrp$ to $\emptyset$ by replacing every letter (group term) in a term $e$ with its reduced form, that is, with the (unique) equivalent group term of minimum size.
For instance, if $\Lambda = \{ \ltr{a}, \ltr{b}, \ltr{c} \}$, then we send the term $\ltr{a} \circ \overline{\ltr{a}} \parallel \ltr{b} \circ \ltr{c} \circ \overline{\ltr{c}}$ to the term $u \parallel \ltr{b}$.
\end{example}

For the remainder of this section, we fix a subalphabet $\Gamma \subseteq \Sigma$.
When $r: \Sigma \to \termscka(\Gamma)$, we extend $r$ to a map from $\termscka(\Sigma)$ to $\termscka(\Gamma)$, by inductively applying $r$ to terms.
We can also apply $r$ to a series-parallel pomset, obtaining a pomset language.
More precisely, when $U$ is a pomset, we define $r(U)$ as follows:
\begin{align*}
r(1) &= \{ 1 \}
    & r(U \cdot V) &= r(U) \cdot r(V) &
r(\ltr{a}) &= \semcka{r(\ltr{a})}
    & r(U \parallel V) &= r(U) \parallel r(V)
\end{align*}
Lastly, when $L \subseteq \SP$, we write $r(L)$ for the set $\bigcup \{ r(U) : U \in L \}$.

The following then formalises the idea of reducing by replacing letters.

\begin{definition}
  A map $r: \Sigma \to \termscka(\Gamma)$ is a \emph{reification} from $H$ to $H'$ if
  \begin{enumerate}[(i)]
  \item\label{property:reify-equivalence-letter}
    For all $\ltr{a} \in \Sigma$, it holds that $r(\ltr{a}) \equivcka^{H} \ltr{a}$.
  \item\label{property:reify-idem}
    $r$ is expansive on $\Gamma$, i.e., for all $\ltr{a} \in\Gamma$, $\ltr{a} \leqq r(\ltr{a})$.

  \item\label{property:reify-preserve-gamma}
    $H'$-closure preserves $\Gamma$, i.e., for all $L \subseteq \SP(\Gamma)$, also $\closure[H']{L} \subseteq \SP(\Gamma)$.

  \item\label{property:reify-hypotheses}
    For all $e\leq f \in H$, it holds that $r(e) \leqqcka^{H'} r(f)$.
  \end{enumerate}
\end{definition}

\begin{example}
Continuing with the previous example, let $r$ be the map that sends a group term to its reduced form; we claim that $r$ is a reification from $\hgrp$ to $\emptyset$.
By definition, we then know that for a group term $g \in \Sigma$, we have $r(g) \equiv_G g$, and hence $r(g) \equiv^\hgrp g$.
Furthermore, the reduction of a reduced term is that term itself; hence, the second condition is satisfied.
The third condition holds trivially.
Lastly, if $e \leq f \in \hgrp$, then $e, f \in \Sigma$ such that $e \equiv_G f$.
Since reductions are unique, we then know that $r(e) = r(f)$, and hence $r(e) \leqqcka^\emptyset r(f)$.
\end{example}

We have the following general properties of a map $r$, which we will use in demonstrating how to obtain a reduction from a reification.
\begin{restatable}{lemma}{contextsem}\label{lemma:reify-context-sem}
Let $r: \Sigma \to \termscka$ be some map.
\begin{enumerate}[(i)]
  \item\label{property:reification-preserves-context}
  For all $C \in \PCsp$, we have $r\left(C\right) \subseteq \PCsp$.

  \item\label{property:reification-vs-context}
  For all $L\subseteq \SP$ and $C \in \PCsp$, we have $r\left(C[L]\right)=\bigcup_{D\in r(C)}D\left[r(L)\right]$.

  \item\label{property:reification-vs-sem}
  For all $e \in \termscka$, it holds that $r(\semcka{e}) = \semcka{r(e)}$.
\end{enumerate}
\end{restatable}

The following technical lemma is a consequence of property~\eqref{property:reify-hypotheses}.

\begin{restatable}{lemma}{propreif}\label{lemma:prop-reif}
  If $r$ is a reification and $L\subseteq\SP(\Sigma)$, then $r(\closure L) \subseteq \closure[H']{r(L)}$.
\end{restatable}

Using this, we can then show how to obtain a reduction from a reification.

\begin{restatable}{lemma}{reificationtoreduction}%
\label{lemma:reification-to-reduction}
  If $H$ implies $H'$ and $r$ is a reification from $H$ to $H'$, then $r$ is a reduction from $H$ to $H'$.
\end{restatable}
\begin{proof}
  The first condition, i.e., that for $e \in \termscka$ we have $e \equivcka^H r(e)$, can be checked using the first property of reification by induction on the structure of $e$.
  It thus remains to check the second condition; we do this by proving that for all $e\in\termscka(\Sigma)$ we have $r\left(\closure {\semcka e}\right)=\closure[H']{\semcka{r(e)}}$.
  To this end, we derive as follows:
  \begin{align*}
  r(\closure{\semcka{e}})
    &\subseteq \closure[H']{r(\semcka{e})}
        \tag{\autoref{lemma:prop-reif}} \\
    &= \closure[H']{\semcka{r(e)}}
        \tag{\autoref{lemma:reify-context-sem}\eqref{property:reification-vs-sem}} \\
    &\subseteq r(\closure[H']{\semcka{r(e)}})
        \tag{property~\eqref{property:reify-idem}} \\
    &\subseteq r(\closure{\semcka{r(e)}})
        \tag{\autoref{lemma:implication}\eqref{property:implication-vs-closure}} \\
    &= r(\closure{\semcka{e}})
        \tag{property~\eqref{property:reify-equivalence-letter}, soundness}
  \end{align*}
  Specifically, in the third step, property~\eqref{property:reify-idem} ensures that for $L \subseteq \SP(\Gamma)$ we have $L \subseteq r(L)$.
  We can use this property because $H'$-closure preserves the $\Gamma$-language by property~\eqref{property:reify-preserve-gamma}.
  This completes the proof.
\end{proof}

\subsection{Factoring the exchange law}%
\label{sec:fact-exch-law}

In the basic axioms that generate $\equivcka$, there is no interaction between sequential and parallel composition.
One sensible way of adding that kind of interaction is, as suggested by Hoare, Struth and collaborators~\cite{hoare-moeller-struth-wehrman-2009}, by adding an axiom of the form $(e \parallel f) \cdot (g \parallel h) \leqqcka (e \cdot g) \parallel (f \cdot h)$, known as the \emph{exchange law}.
Essentially, this axiom encodes the possibility of (partial) interleaving: when $e \cdot g$ runs in parallel with $f \cdot h$, one possible behaviour is that, first $e$ runs in parallel with $f$, and then $g$ runs in parallel with $h$.
The core observation of this section is that the exchange law can be treated as another set of hypotheses, as we show below, and this can then be used to recover the completeness result of CKA~\cite{kappe-brunet-silva-zanasi-2018}.

\begin{definition}
We write $\hexch$ for the set
\[
    \{ (e \parallel f) \cdot (g \parallel h) \leq (e \cdot g) \parallel (f \cdot h) : e, f, g, h \in \termscka \}
\]
\end{definition}

The semantic effect of adding $\hexch$ to our hypotheses is that, if $U$ is a pomset in a series-parallel language $L$, and $V$ is a series-parallel pomset subsumed by $U$, then $V$ is in the $\hexch$-closure of $L$.
Intuitively, the $\hexch$-closure adds pomsets that are more sequential, i.e., have more ordering, than the ones already in $L$.
Indeed, $\hexch$-closure coincides with the downward closure w.r.t. $\subsp$.

\begin{restatable}{lemma}{exchclosure}%
\label{lemma:exch-closure-vs-subsumption}
Let $L \subseteq \SP$ and $U \in \SP$.
Now $U\in\closure[\hexch]L$ if and only if there exists a $V \in L$ such that $U \subsp V$.
\end{restatable}

We have previously shown that $\hexch$ is complete~\cite{kappe-brunet-silva-zanasi-2018}; as a matter of fact, the pivotal result from op.\ cit.\ can be presented as follows.
\begin{theorem}%
\label{theorem:reduce-exch}
The set of hypotheses $\hexch$ is strongly reducible to $\emptyset$.
\end{theorem}

When $\hexch$ is contained in our hypotheses, it is not immediately clear whether those hypotheses can be reduced.
What we can do is try to factorise our hypotheses into $\hexch$ and some residual set of hypotheses, and prove strong reducibility for that residual set.
To this end, we first note that, in some circumstances, the $H$-closure of the $\hexch$-closure remains downward-closed w.r.t. $\subsp$.

\begin{restatable}{lemma}{factoriseseq}%
\label{lemma:factorise-seq-exch}
Suppose that for each $e \leq f \in H$ we have that $e = 1$ or $e = \ltr{a}$ for some $\ltr{a} \in \Sigma$, and let $L \subseteq \SP$.
If $U, V \in \SP$ such that $U \subsp V$ and $V \in \closure{(\closure[\hexch]L)}$, then $U\in\closure{(\closure[\hexch]L)}$.
\end{restatable}

Using this fact, we can now show that, under the same precondition, $\hexch \cup H$ factors into $\hexch$ and $H$.
This factorisation is what we were looking for: it tells us that whenever $H$ strongly reduces to $\emptyset$, so does $H \cup \hexch$.

\begin{restatable}{lemma}{factorix}%
\label{lemma:factorise-exch}
Suppose that for each $e \leq f \in H$ we have that $e = 1$, or $e = \ltr{a}$ for some $\ltr{a} \in \Sigma$.
Then $H \cup \hexch$ factorises into $\hexch$ and $H$.
\end{restatable}
\begin{proof}
Since $H, \hexch \subseteq H \cup \hexch$, it should be obvious that $H \cup \hexch$ implies both $H$ and $\hexch$.
It remains to show that, if $L \subseteq \SP$, then $\closure[H]{(\closure[\hexch]{L})} = \closure[H \cup \hexch]{L}$.
The inclusion from left to right is a consequence of \autoref{lemma:implication}\eqref{property:implication-vs-closure}--\eqref{property:implication-vs-double-closure}.

For the other inclusion, we show that if $A \subseteq \closure[H \cup \hexch]{L}$, then $A \subseteq \closure[H]{(\closure[\hexch]{L})}$.
The proof proceeds by induction on the construction of $A \subseteq \closure[H \cup \hexch]L$.
In the base, we have that $A \subseteq \closure[H \cup \hexch]{L}$ because $A = L$; in that case, $A \subseteq \closure[\hexch]{L} \subseteq \closure[H]{(\closure[\hexch]{L})}$.

For the inductive step, $A \subseteq \closure[H \cup \hexch]{L}$ because there exist $e \leq f \in H \cup \hexch$ and $C \in \PCsp$ such that $A = C[\semcka{e}]$, and $C[\semcka{f}] \subseteq \closure[H \cup \hexch]{L}$.
By induction, we then know that $C[\semcka{f}] \subseteq \closurep{\closure[\hexch]L}$.
On the one hand, if $e \leq f \in H$, then $A = C[\semcka{e}] \subseteq \closurep{\closure[\hexch]L}$ immediately.
On the other hand, if $e \leq f \in \hexch$, then $\semcka{e} \subsp \semcka{f}$, and hence $C[\semcka{e}] \subsp C[\semcka{f}]$ by \autoref{lemma:pomset-context-monotone}.
By \autoref{lemma:context-preserve-series-parallel} and \autoref{lemma:factorise-seq-exch}, it then follows that $A = C[\semcka{e}] \subseteq \closure[H]{(\closure[\hexch]{L})}$.
\end{proof}

\subsection{Lifting}

A number of reduction procedures already exist at the level of Kleene algebra~\cite{kozen-mamouras-2014,doumane-kuperberg-pous-pradic-2019}; ideally, one would like to lift those procedures to CKA\@.

\begin{example}
The reductions in \autoref{example:simple-reduction} and \autoref{example:reify-group} worked out for terms without $\parallel$, and then extended inductively, by defining the reduction of $e \parallel f$ to be the parallel composition of the reductions of $e$ and $f$ respectively.

As a non-example, consider $H = \{ \ltr{a} \leq 1 \}$.
Even though this hypothesis can be reduced to $\emptyset$ within Kleene algebra~\cite{cohen-1994}, it is not obvious how this would work for pomset languages.
In particular, if $1 \in L$, then $1 \parallel \dots \parallel 1 \in L$ for any number of $1$'s, and hence $\ltr{a} \parallel \dots \parallel \ltr{a} \in \closure{L}$ for any number of $\ltr{a}$'s.
This precludes the possibility of a strong reduction to $\emptyset$, because $\closure{\semcka{1}}$ is a pomset language of unbounded (parallel) width, which cannot be expressed by any $e \in \termscka$~\cite{lodaya-weil-2000}.
\end{example}

We now establish a set of sufficient conditions for such a lifting to work.
To this end, we first formally define Kleene algebra syntax, axioms and semantics.

\begin{definition}
Write $\termska$ for the set of \emph{Kleene algebra terms}, i.e., the terms in $\termscka$ that do not contain $\parallel$.
Furthermore, we write $\equivka$ for the smallest congruence on $\termska$ that is generated by the axioms of $\equivcka$ that do not involve $\parallel$.
\end{definition}

When $e \in \termska$, it is not hard to see that $\semcka{e}$ contains totally ordered pomsets, i.e., words, exclusively.
Using these definitions, we can now specialise the notions of hypotheses, context, and closure to the sequential setting, as follows.

\begin{definition}
The relation $\equivka^H$ is generated from $H$ and $\equivka$ as before.

A context $C \in \PCsp$ is \emph{sequential} if it is totally ordered, i.e., if it is a word with one occurrence of $*$; we write $\PCseq$ for the set of sequential contexts.

Given a set of hypotheses $H$ and a language $L \subseteq \Sigma^*$, we define the \emph{sequential closure} of $L$ with respect to $H$, written $\seqclosure L$, as the least language containing $L$ such that for all $e \leq f \in H$ and $C \in \PCseq$, if $C[\semcka f]\subseteq \seqclosure L$, then $C[\semcka e]\subseteq\seqclosure L$.
\end{definition}

If $\parallel$ does not occur in any hypothesis, then the definition of sequential closure coincides with the closure operator from~\cite{doumane-kuperberg-pous-pradic-2019}.
Thus, if $L \subseteq \Sigma^*$, then $\seqclosure{L} \subseteq \Sigma^*$.

The analogue of strong reduction for the sequential setting is as follows.

\begin{definition}
Suppose that $H$ implies $H'$.
A map $r: \termska \to \termska$ is a \emph{sequential reduction} from $H$ to $H'$ when the following hold:
\begin{enumerate}[(i)]
    \item\label{property:seq-preserve-equivalence}
    for $e \in \termska$, it holds that $e \equivka^H r(e)$, and

    \item\label{property:seq-project-equivalence}
    for $e \in \termska$, it holds that $\seqclosure[H]{\semka{e}} = \seqclosure[H']{\semka{r(e)}}$.
\end{enumerate}
$H$ \emph{sequentially reduces} to $H'$ if there exists a sequential reduction from $H$ to $H'$.
\end{definition}

To lift a sequential reduction to a proper reduction, the following class of hypotheses will turn out to be useful.

\begin{definition}
A hypothesis $e \leq f$ with $e, f \in \termska$ is called \emph{grounded} if $\semcka{f} = \{ W \}$ for some non-empty word (totally ordered pomset) $W$, and $e \in \termska$.
We say that a set of hypotheses $H$ is grounded if every $e \leq f \in H$ is grounded.
\end{definition}

\begin{example}
Any hypothesis of the form $e \leq \ltr{a}_1 \cdots \ltr{a}_n$ for $n > 0$ is grounded.
On the other hand, the hypothesis $\ltr{a} \leq 1$ that we saw in the previous example is not grounded, since the semantics of $1$ contains the empty pomset.
\end{example}

The closure of a language of words can be expressed in terms of its sequential closure, provided that the set of hypotheses is grounded.

\begin{restatable}{lemma}{liftclosure}%
\label{lemma:lift-closure}
Let $H$ be grounded.
If $L \subseteq \Sigma^*$, then $\closure{L} = \seqclosure{L}$.
Moreover, for $L, L' \subseteq \SP$, we have that $\closurep{L \parallel L'} = \closure{L} \parallel \closure{L'}$.
\end{restatable}

The above then allows us to turn a sequential reduction into a reduction.

\begin{restatable}{lemma}{liftreduction}%
\label{lemma:lift-reduction}
Suppose that $H$ sequentially reduces to $H'$.
If $H$ and $H'$ are grounded, then $H$ strongly reduces to $H'$.
\end{restatable}

\section{Instantiation to CKA with Observations}%
\label{section:instantiationtockao}
In this section, we will present Concurrent Kleene Algebra with Observations (CKAO), an extension of CKA with Boolean assertions that enable the specification of programs with the usual guarded conditionals and loops. We will obtain CKAO as an instance of CKAH by choosing a particular set of hypotheses. First, we define the set of propositional terms or Boolean observations.
\begin{definition}
Fix a finite set $\Omega$ of \emph{primitive observations}.
The set of \emph{propositional terms}, written $\termsba$, is generated by
\[
    p, q ::= \bot \pipe \top \pipe o \in \Omega \pipe p \vee q \pipe p \wedge q \pipe \overline{p}
\]
The relation $\equivba$ is the smallest congruence on $\termsba$ s.t.\ for $p, q, r \in \termsba$, we have
\begin{mathpar}
p \vee \bot \equivba p
\and
p \vee q \equivba q \vee p
\and
p \vee \overline{p} \equivba \top
\and
p \vee (q \vee r) \equivba (p \vee q) \vee r
\\
p \wedge \top \equivba p
\and
p \wedge q \equivba q \wedge p
\and
p \wedge \overline{p} \equivba \bot
\and
p \wedge (q \wedge r) \equivba (p \wedge q) \wedge r
\\
p \vee (q \wedge r) \equivba (p \vee q) \wedge (p \vee r)
\and
p \wedge (q \vee r) \equivba (p \wedge q) \vee (p \wedge r)
\end{mathpar}
We will write $p \leqqba q$ as a shorthand for $p \vee q \equivba q$.
\end{definition}
We write $\At$ for $2^\Omega$, the set of \emph{atoms} of the Boolean algebra.
It is well known that every $\alpha \in \At$ corresponds canonically to a Boolean term $\pi_\alpha$, such that every Boolean term $p \in \termsba$ is equivalent to the disjunction of all $\pi_\alpha$ with $\pi_\alpha \leqqba p$~\cite{birkhoff-bartee-1970}.
To simplify notation we identify $\alpha \in \At$ with $\pi_\alpha$.

We can now use $\termsba$ in defining the terms and axioms of CKAO, which will be given as a CKA over a specific alphabet with the following hypotheses:

\begin{definition}[CKAO]
We define the \emph{terms} of CKAO, denoted $\termsckao$, as $\termscka(\Sigma \cup \termsba)$, that is, as the CKA terms over $\termsba \cup \Sigma$.
We furthermore define the following set of hypotheses over $\termsckao$:
\begin{mathpar}
\hbool = \{ p = q : p, q \in \termsba\ \mathrm{s.t.}\ p \equivba q \}
\and
\hcontr = \{ p \wedge q \leq p \cdot q : p, q \in \termsba \}
\and
\hglu = \{ 0 = \bot \} \cup \{ p + q = p \vee q : p, q \in \termsba \}
\and
\hobs = \hbool \cup \hcontr \cup \hexch \cup \hglu
\end{mathpar}
The \emph{semantics} of CKAO is then given by $\closure[\hobs]{\semcka{-}}$.
\end{definition}
The hypotheses $\hbool$ contain the boolean identities, and $\hglu$ identifies the disjunction with the union (and their respective units as well).
$\hcontr$ specifies that if $p$ and $q$ hold simultaneously, then it is possible to observe them in sequence.
Note that the converse inequality is not included: observing $p$ and $q$ in sequence has strictly more behaviour than observing $p$ and $q$ simultaneously, as some intervening action can happen between the two observations.

The above definition gives us the semantics of CKAO as the standard pomset language model obtained from taking the $\hobs$-closure of the semantics of CKA\@.
As a matter of fact, we find by \autoref{lemma:soundness} that if $e, f \in \termsckao$ with $e \equivcka^\hobs f$, then $\closure[\hobs]{\semcka{e}} = \closure[\hobs]{\semcka{f}}$; hence, we already have a sound model of CKAO\@.

To prove completeness, we will use the techniques from the previous section.

\paragraph{First step: reification.}
We start by using reification to rid ourselves of the hypotheses from $\hbool$ and $\hglu$, and to simplify the hypotheses in $\hcontr$.
To this end, let $\hcontr'$ be the set of hypotheses given by $\{ \alpha \leq \alpha \cdot \alpha : \alpha \in \At \}$.
Let $\Gamma = \At \cup \Sigma \subseteq \termsba \cup \Sigma$.
We define $r: \Sigma \cup \termsba \to \termscka(\Gamma)$ by setting
\[
    r(a) =
        \begin{cases}
        \sum_{\alpha \leqqba p} \alpha & a = p \in \termsba \\
        \ltr{a} & a = \ltr{a} \in \Sigma
        \end{cases}
\]

\begin{lemma}%
\label{lemma:instantiate-reification}
The hypotheses $\hobs$ reduce to $\hexch \cup \hcontr'$.
\end{lemma}
\begin{proof}
By \autoref{lemma:reification-to-reduction}, it suffices to show that $r$ is a reification, and that $\hobs$ implies $\hexch \cup \hcontr'$.
To see that $r$ is a reification, we check the conditions.

\eqref{property:reify-equivalence-letter}: 
If $\ltr{a} \in \Sigma$, then $r(\ltr{a}) = \ltr{a} \equivcka^\hobs \ltr{a}$ immediately.
Otherwise, if $p \in \termsba$, then we derive $r(p) = \sum_{\alpha \leqqba p} \alpha \equivcka^\hglu \bigvee_{\alpha \leqqba p} \alpha \equivcka^\hbool p$
and hence $r(p) \equivcka^\hobs p$.

\eqref{property:reify-idem}: 
If $\ltr{a} \in \Sigma$, then we already know that $r(\ltr{a}) = \ltr{a}$.
Otherwise, if $\alpha \in \At$, then
\[
  r(\alpha) = \sum_{\beta \leqqba \alpha} \beta = \alpha
\]

\eqref{property:reify-preserve-gamma}: 
This property holds because all hypotheses in $\hexch \cup \hcontr'$ preserve $\Gamma$-languages, i.e., if $e \leq f \in \hexch \cup \hcontr'$ where $\semcka{f} \subseteq \SP(\Gamma)$, then $\semcka{e} \subseteq \SP(\Gamma)$ too.
It follows that $\hexch \cup \hcontr'$-closure must preserve $\Gamma$-languages.

\eqref{property:reify-hypotheses}: 
We should show that if $e \leq f \in \hobs$, then $r(e) \leqqcka^{\hexch \cup \hcontr'} r(f)$.
To this end, we analyse the separate sets of hypotheses that make up $\hobs$.
\begin{itemize}
\item
  Let $e \leq f \in \hexch$, then
  $
  e = (g_{00} \parallel g_{01}) \cdot (g_{10} \parallel g_{11})$
  and
  $ f = (g_{00} \cdot g_{10}) \parallel (g_{01} \cdot g_{11})$,
  for some $g_{00}, g_{01}, g_{10}, g_{11} \in \termscka$.
  We then find that
  \begin{mathpar}
    r(e) = (r(g_{00}) \parallel r(g_{01})) \cdot (r(g_{10}) \parallel r(g_{11}))
    \and
    r(f) = (r(g_{00}) \cdot r(g_{10})) \parallel (r(g_{01}) \cdot r(g_{11}))
  \end{mathpar}
  hence $r(e) \leq r(f) \in \hexch$, and therefore $r(e) \leqqcka^{\hexch \cup \hcontr'} r(f)$.

\item
  Let $e \leq f \in \hbool$, then $e = p$ and $f = q$ such that $p \equivba q$.
  In that case,
  \[
    r(p) = \sum_{\alpha \leqqba p} \alpha = \sum_{\alpha \leqqba q} \alpha = r(q)
  \]

\item
  Let $e \leq f \in \hcontr$; then $e = p \wedge q$ and $f = p \cdot q$ for $p, q \in \termsba$.
  Then
  \begin{align*}
    r(p \wedge q)
    &= \sum_{\alpha \leqqba p \wedge q} \alpha
      \leqqcka^{\hcontr'} \sum_{\alpha \leqqba p \wedge q} \alpha \cdot \alpha \\
    &\leqqcka \Bigl( \sum_{\alpha \leqqba p} \alpha \Bigr) \cdot \Bigl( \sum_{\alpha \leqqba q} \alpha \Bigr)
      = r(p) \cdot r(q)
      = r(p \cdot q)
  \end{align*}

\item
  Let $e \leq f \in \hglu$.
  On the one hand, if $e = p \vee q$ and $f = p + q$, then
  \[
    r(p \vee q)
    = \sum_{\alpha \leqqba p \vee q} \alpha
    \equivcka \sum_{\alpha \leqqba p} \alpha + \sum_{\alpha \leqqba q} \alpha
    = r(p) + r(q)
    = r(p + q)
  \]
  This also establishes the case for $f \leq e \in \hglu$.
  On the other hand, if $e = 0$ and $p = \bot$, then
  $
    r(0)
    = 0
    = \sum_{\alpha \leqqba \bot} \alpha
    = r(\bot)
  $.
\end{itemize}

\noindent
To see that $\hobs$ implies $\hexch \cup \hcontr'$, it suffices to show that $\hobs$ implies $\hcontr'$.
To this end, note that if $e \leq f \in \hcontr'$, then $e = \alpha$ and $f = \alpha \cdot \alpha$ for some $\alpha \in \At$.
We can then derive that $\alpha \equivcka^\hbool \alpha \wedge \alpha \leqqcka^\hcontr \alpha \cdot \alpha$, and hence $e \leqqcka^\hobs f$.
\end{proof}

\paragraph{Second step: factorising.}

Since $\hcontr'$ satisfies the precondition of \autoref{lemma:factorise-exch}, we obtain the following.
\begin{lemma}%
\label{lemma:instantiate-factorisation}
The hypotheses $\hexch \cup \hcontr'$ factorise into $\hexch$ and $\hcontr'$.
\end{lemma}

This means that, by \autoref{lemma:factorisation} all that remains to do is strongly reduce $\hexch$ and $\hcontr'$ to $\emptyset$; we have already taken care of the former in \autoref{theorem:reduce-exch}.

\paragraph{Third step: reducing $\hcontr'$.}

In~\cite{kappe-brunet-rot-silva-wagemaker-zanasi-2019}, we have already shown that $\hcontr'$ sequentially reduces to $\emptyset$.
Since $\hcontr'$ is grounded we find the following, by \autoref{lemma:lift-reduction}.
\begin{lemma}%
\label{lemma:reduce-contrprime}
The hypotheses $\hcontr'$ strongly reduce to $\emptyset$.
\end{lemma}

\paragraph{Last step: putting it all together.}

Using the above reductions, we can then prove completeness of $\equivcka^\hobs$ w.r.t. $\closure[\hobs]{\semcka{-}}$, and decidability of semantic equivalence, too.

\begin{theorem}[Soundness and Completeness of CKAO]
Let $e, f \in \termsckao$.
\begin{enumerate}[(i)]
    \item
    We have $e \equiv^\hobs f$ if and only if $\closure[\hobs]{\semcka{e}} = \closure[\hobs]{\semcka{f}}$.

    \item
    It is decidable whether $\closure[\hobs]{\semcka{e}} = \closure[\hobs]{\semcka{f}}$.
\end{enumerate}
\end{theorem}
\begin{proof}
For the first claim, we already knew the implication from left to right from \autoref{lemma:soundness}.
Conversely, and for the second claim, first note that that $\hobs$ reduces to $\hexch \cup \hcontr'$ by \autoref{lemma:instantiate-reification}.
By \autoref{lemma:instantiate-factorisation} and \autoref{lemma:factorisation}, the latter reduces to $\emptyset$, if we apply \autoref{theorem:reduce-exch} and \autoref{lemma:reduce-contrprime}.
By \autoref{lemma:carry-over}, we then conclude that $\hobs$ is complete and decidable, hence establishing the claim.
\end{proof}

\section{Discussion}%
\label{section:discussion}
The first contribution of this paper is to extend Kleene algebra with hypotheses~\cite{doumane-kuperberg-pous-pradic-2019} with a parallel operator.
The resulting framework, concurrent Kleene algebra with hypotheses (CKAH), is interpreted over pomset languages, a standard model of concurrency.
We start from simple axioms, known to capture equality of pomset languages~\cite{laurence-struth-2014}.
CKAH allows to add custom axioms, the so-called hypotheses.
These may be used to include domain-specific information in the language.
We develop this framework by providing a systematic way of producing from the hypotheses a sound pomset language model.
We also propose techniques that may be used to prove completeness and decidability of the resulting model.

An important instance of this framework is concurrent Kleene algebra (CKA) as presented in~\cite{hoare-moeller-struth-wehrman-2009}.
The only additional axiom there, known as the exchange law, may be added as a set of hypotheses.
We prove that the resulting semantics coincides with the (subsumption-closed) semantics of CKA and, more interestingly, the completeness proof of~\cite{kappe-brunet-silva-zanasi-2018} can be recovered as an instance of this framework.

The second contribution is a new framework to reason about programs with concurrency: concurrent Kleene algebra with observations (CKAO).
CKAO is obtained as an instance of CKAH, where we add the exchange law to model concurrent behaviour, and Boolean assertions to model control flow.
The Boolean assertions we consider are as in Kleene algebra with observations (KAO)~\cite{kappe-brunet-rot-silva-wagemaker-zanasi-2019} --- in fact, CKAO is a conservative extension of KAO\@.
Using the techniques developed earlier, we obtain a sound and complete semantics for this algebra.
While CKAO is similar to concurrent Kleene algebra with tests~\cite{jipsen-moshier-2016}, it avoids the problems of the latter by distinguishing conjunction and sequential composition.
CKAO provides the first sound and complete algebraic theory that seems sensible as a framework to reason about concurrent programs with Boolean assertions.

Future work is to explore other meaningful instances of CKAH\@.
Synchronous Kleene algebra~\cite{wagemaker-bonsangue-etal-2019,prisacariu-2010} is a natural candidate for this.
We also want to try and design domain specific languages, specifically, a concurrent variant of NetKAT~\cite{anderson-et-al-2014,foster-kozen-etal-2015}.

The class of hypotheses considered in this paper for which decidability and completeness may be established systematically is somewhat restrictive; identifying larger classes of tractable hypotheses is a challenging open problem.

Because of the compositional nature of our model, the CKAO semantics of a program contains behaviours that are not possible to obtain in isolation.
These behaviours are present to allow the program to interact meaningfully with its environment, i.e., when placed in a context.
However, for practical purposes one might want to close the system, and only consider behaviours that are possible in isolation. Studying this semantics remains subject of future work.

In the semantics of concurrent programs with assertions, it would be natural to see atoms as partial instead of total functions.
This captures the intuition that a thread might not have access to the complete machine state, but instead holds a partial view of it.
Pseudo-complemented distributive lattices (PCDL) have been proposed~\cite{jipsen-moshier-2016} as an alternative to Boolean algebra, modelling this partiality of information.
We leave it to future work to investigate the variant of CKAO obtained by replacing the Boolean algebra of observations with a PCDL\@.

\clearpage
\bibliography{bibliography}

\begin{thebibliography}{10}

\bibitem{anderson-et-al-2014}
Carolyn~Jane Anderson, Nate Foster, Arjun Guha, Jean-Baptiste Jeannin, Dexter
  Kozen, Cole Schlesinger, and David Walker.
\newblock {NetKAT}: Semantic foundations for networks.
\newblock In {\em POPL}, pages 113--126. ACM, 2014.
\newblock \href {https://doi.org/10.1145/2535838.2535862}
  {\path{doi:10.1145/2535838.2535862}}.

\bibitem{birkhoff-bartee-1970}
Garrett Birkhoff and Thomas~C. Bartee.
\newblock {\em Modern applied algebra}.
\newblock McGraw-Hill, 1970.

\bibitem{bonchi-pous-2013}
Filippo Bonchi and Damien Pous.
\newblock Checking {NFA} equivalence with bisimulations up to congruence.
\newblock In {\em POPL}, pages 457--468, 2013.
\newblock \href {https://doi.org/10.1145/2429069.2429124}
  {\path{doi:10.1145/2429069.2429124}}.

\bibitem{brunet-pous-struth-2017}
Paul Brunet, Damien Pous, and Georg Struth.
\newblock On decidability of concurrent {K}leene algebra.
\newblock In {\em CONCUR}, pages 28:1--28:15, 2017.
\newblock \href {https://doi.org/10.4230/LIPIcs.CONCUR.2017.28}
  {\path{doi:10.4230/LIPIcs.CONCUR.2017.28}}.

\bibitem{cohen-1994}
Ernie Cohen.
\newblock Hypotheses in {K}leene algebra.
\newblock Technical report, Bellcore, 1994.

\bibitem{conway-1971}
John~Horton Conway.
\newblock {\em Regular Algebra and Finite Machines}.
\newblock Chapman and Hall, Ltd., London, 1971.

\bibitem{doumane-kuperberg-pous-pradic-2019}
Amina Doumane, Denis Kuperberg, Damien Pous, and Pierre Pradic.
\newblock Kleene algebra with hypotheses.
\newblock In {\em FOSSACS}, pages 207--223, 2019.
\newblock \href {https://doi.org/10.1007/978-3-030-17127-8_12}
  {\path{doi:10.1007/978-3-030-17127-8_12}}.

\bibitem{foster-kozen-etal-2015}
Nate Foster, Dexter Kozen, Matthew Milano, Alexandra Silva, and Laure Thompson.
\newblock A coalgebraic decision procedure for {NetKAT}.
\newblock In {\em POPL}, pages 343--355, 2015.
\newblock \href {https://doi.org/10.1145/2676726.2677011}
  {\path{doi:10.1145/2676726.2677011}}.

\bibitem{gischer-1988}
Jay~L. Gischer.
\newblock The equational theory of pomsets.
\newblock {\em Theor. Comput. Sci.}, 61:199--224, 1988.
\newblock \href {https://doi.org/10.1016/0304-3975(88)90124-7}
  {\path{doi:10.1016/0304-3975(88)90124-7}}.

\bibitem{grabowski-1981}
Jan Grabowski.
\newblock On partial languages.
\newblock {\em Fundam. Inform.}, 4(2):427, 1981.

\bibitem{hoare-moeller-struth-wehrman-2009}
Tony Hoare, Bernhard M{\"{o}}ller, Georg Struth, and Ian Wehrman.
\newblock Concurrent {K}leene algebra.
\newblock In {\em CONCUR}, pages 399--414, 2009.
\newblock \href {https://doi.org/10.1007/978-3-642-04081-8_27}
  {\path{doi:10.1007/978-3-642-04081-8_27}}.

\bibitem{jipsen-moshier-2016}
Peter Jipsen and M.~Andrew Moshier.
\newblock Concurrent {K}leene algebra with tests and branching automata.
\newblock {\em J. Log. Algebr. Meth. Program.}, 85(4):637--652, 2016.
\newblock \href {https://doi.org/10.1016/j.jlamp.2015.12.005}
  {\path{doi:10.1016/j.jlamp.2015.12.005}}.

\bibitem{kappe-brunet-rot-silva-wagemaker-zanasi-2019}
Tobias Kapp{\'{e}}, Paul Brunet, Jurriaan Rot, Alexandra Silva, Jana Wagemaker,
  and Fabio Zanasi.
\newblock Kleene algebra with observations.
\newblock In {\em CONCUR}, pages 41:1--41:16, 2019.
\newblock \href {https://doi.org/10.4230/LIPIcs.CONCUR.2019.41}
  {\path{doi:10.4230/LIPIcs.CONCUR.2019.41}}.

\bibitem{kappe-brunet-silva-zanasi-2018}
Tobias Kapp{\'{e}}, Paul Brunet, Alexandra Silva, and Fabio Zanasi.
\newblock Concurrent {K}leene algebra: Free model and completeness.
\newblock In {\em ESOP}, pages 856--882, 2018.
\newblock \href {https://doi.org/10.1007/978-3-319-89884-1_30}
  {\path{doi:10.1007/978-3-319-89884-1_30}}.

\bibitem{kozen-1994}
Dexter Kozen.
\newblock A completeness theorem for {K}leene algebras and the algebra of
  regular events.
\newblock {\em Inf. Comput.}, 110(2):366--390, 1994.
\newblock \href {https://doi.org/10.1006/inco.1994.1037}
  {\path{doi:10.1006/inco.1994.1037}}.

\bibitem{kozen-1996}
Dexter Kozen.
\newblock Kleene algebra with tests and commutativity conditions.
\newblock In {\em TACAS}, pages 14--33, 1996.
\newblock \href {https://doi.org/10.1007/3-540-61042-1_35}
  {\path{doi:10.1007/3-540-61042-1_35}}.

\bibitem{kozen-2002}
Dexter Kozen.
\newblock On the complexity of reasoning in {K}leene algebra.
\newblock {\em Inf. Comput.}, 179(2):152--162, 2002.
\newblock \href {https://doi.org/10.1006/inco.2001.2960}
  {\path{doi:10.1006/inco.2001.2960}}.

\bibitem{kozen-2017}
Dexter Kozen.
\newblock On the coalgebraic theory of {K}leene algebra with tests.
\newblock In Can Ba{\c{s}}kent, Lawrence~S. Moss, and Ramaswamy Ramanujam,
  editors, {\em Rohit Parikh on Logic, Language and Society}, volume~11 of {\em
  Outstanding Contributions to Logic}, pages 279--298. Springer, 2017.

\bibitem{kozen-mamouras-2014}
Dexter Kozen and Konstantinos Mamouras.
\newblock Kleene algebra with equations.
\newblock In {\em ICALP}, pages 280--292, 2014.
\newblock \href {https://doi.org/10.1007/978-3-662-43951-7_24}
  {\path{doi:10.1007/978-3-662-43951-7_24}}.

\bibitem{krob-1990}
Daniel Krob.
\newblock A complete system of {B}-rational identities.
\newblock In {\em ICALP}, pages 60--73, 1990.
\newblock \href {https://doi.org/10.1007/BFb0032022}
  {\path{doi:10.1007/BFb0032022}}.

\bibitem{kuratowski-1922}
Casimir Kuratowski.
\newblock {Sur l'op{\'e}ration {\=A} de l'Analysis Situs}.
\newblock {\em Fundamenta Mathematicae}, 3(1):182--199, 1922.

\bibitem{laurence-struth-2014}
Michael~R. Laurence and Georg Struth.
\newblock Completeness theorems for bi-{K}leene algebras and series-parallel
  rational pomset languages.
\newblock In {\em RAMiCS}, pages 65--82, 2014.
\newblock \href {https://doi.org/10.1007/978-3-319-06251-8_5}
  {\path{doi:10.1007/978-3-319-06251-8_5}}.

\bibitem{laurence-struth-2017-arxiv}
Michael~R. Laurence and Georg Struth.
\newblock Completeness theorems for pomset languages and concurrent {K}leene
  algebras, 2017.
\newblock \href {http://arxiv.org/abs/abs/1705.05896}
  {\path{arXiv:abs/1705.05896}}.

\bibitem{lodaya-weil-2000}
Kamal Lodaya and Pascal Weil.
\newblock Series-parallel languages and the bounded-width property.
\newblock {\em Theoretical Computer Science}, 237(1):347--380, 2000.
\newblock \href {https://doi.org/10.1016/S0304-3975(00)00031-1}
  {\path{doi:10.1016/S0304-3975(00)00031-1}}.

\bibitem{prisacariu-2010}
Cristian Prisacariu.
\newblock Synchronous {K}leene algebra.
\newblock {\em The Journal of Logic and Algebraic Programming}, 79(7):608 --
  635, 2010.
\newblock \href {https://doi.org/10.1016/j.jlap.2010.07.009}
  {\path{doi:10.1016/j.jlap.2010.07.009}}.

\bibitem{salomaa-1966}
Arto Salomaa.
\newblock Two complete axiom systems for the algebra of regular events.
\newblock {\em J. {ACM}}, 13(1):158--169, 1966.
\newblock \href {https://doi.org/10.1145/321312.321326}
  {\path{doi:10.1145/321312.321326}}.

\bibitem{smolka-foster-etal-2020}
Steffen Smolka, Nate Foster, Justin Hsu, Tobias Kapp{\'{e}}, Dexter Kozen, and
  Alexandra Silva.
\newblock Guarded {K}leene algebra with tests: verification of uninterpreted
  programs in nearly linear time.
\newblock In {\em POPL}, pages 61:1--61:28, 2020.
\newblock \href {https://doi.org/10.1145/3371129} {\path{doi:10.1145/3371129}}.

\bibitem{wagemaker-bonsangue-etal-2019}
Jana Wagemaker, Marcello Bonsangue, Tobias Kappé, Jurriaan Rot, and Alexandra
  Silva.
\newblock Completeness and incompleteness of synchronous {K}leene algebra.
\newblock In {\em MPC}, 2019.
\newblock \href {https://doi.org/10.1007/978-3-030-33636-3_14}
  {\path{doi:10.1007/978-3-030-33636-3_14}}.

\end{thebibliography}

\vfill

{\small\medskip\noindent{\bf Open Access} This chapter is licensed under the terms of the Creative Commons\break Attribution 4.0 International License (\url{http://creativecommons.org/licenses/by/4.0/}), which permits use, sharing, adaptation, distribution and reproduction in any medium or format, as long as you give appropriate credit to the original author(s) and the source, provide a link to the Creative Commons license and indicate if changes were made.} 

{\small \spaceskip .28em plus .1em minus .1em The images or other third party material in this chapter are included in the chapter's Creative Commons license, unless indicated otherwise in a credit line to the material.~If material is not included in the chapter's Creative Commons license and your intended\break use is not permitted by statutory regulation or exceeds the permitted use, you will need to obtain permission directly from the copyright holder.} 

\medskip\noindent\includegraphics{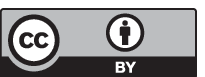} 

\appendix
\ifarxiv%

\clearpage
\section{Omitted Proofs}%
\label{appendix:proofs}

\pomsetcontextmonotone*
\begin{proof}
Let $C = [\lp{c}]$, $D = [\lp{d}]$ and $U = [\lp{u}]$.
Without loss of generality, we can assume that $S_\lp{c} = S_{\lp{d}}$ and $\lambda_\lp{c} = \lambda_{\lp{d}}$ and ${\leq_\lp{d}} \subseteq {\leq_{\lp{c}}}$.
We can furthermore assume without loss of generality that $S_\lp{c}$ is disjoint from $S_\lp{u}$.
We write $C[U] = [\lp{c}[\lp{u}]]$ and $D[U] = [\lp{d}[\lp{u}]]$, and note that $S_{\lp{c}[\lp{u}]} = S_{\lp{d}[\lp{u}]}$ as well as
$\lambda_{\lp{c}[\lp{u}]} = \lambda_{\lp{d}[\lp{u}]}$ by definition.
The claim is then established by showing that ${\leq_{\lp{d}[\lp{u}]}} \subseteq {\leq_{\lp{c}[\lp{u}]}}$.
To this end, suppose that $s, s' \in S_{\lp{d}[\lp{u}]}$ such that $s \leq_{\lp{d}[\lp{u}]} s'$; there are four cases.
\begin{itemize}
    \item
    If $s \leq_\lp{u} s'$, then $s \leq_{\lp{c}[\lp{u}]} s'$ by definition.

    \item
    If $s \leq_\lp{d} s'$, then $s \leq_{\lp{c}} s'$, because ${\leq_\lp{d}} \subseteq {\leq_{\lp{c}}}$, and hence $s \leq_{\lp{c}[\lp{u}]} s'$ by definition.

    \item
    If $s \in S_\lp{u}$ and $s_* \leq_\lp{d} s'$, then $s_* \leq_{\lp{c}} s$, and hence $s \leq_{\lp{c}[\lp{u}]} s'$ by definition.

    \item
    If $s' \in S_\lp{u}$ and $s \leq_\lp{d} s_*$, then $s \leq_{\lp{c}} s_*$, and hence $s \leq_{\lp{c}[\lp{u}]} s'$ by definition.
    \qedhere
\end{itemize}
\end{proof}

\seriesparallelcontextsinductive*
\begin{proof}
Let $L$ be the smallest pomset language satisfying the rules above.

To see that $L \subseteq \PCsp$, it suffices to show that $\PCsp$ satisfies the same rules.
For the first rule, we have that $* \in \PCsp$ because $*$ is a pomset with exactly one $*$-labelled element, and $*$ is series-parallel because it is primitive.
For the second rule, suppose $U \in \SP$ and $V \in \PCsp$.
Then $U$ and $V$ are series-parallel, and hence $U \cdot V$ must be series-parallel as well.
Furthermore, since $U$ has no $*$-labelled nodes (on account of $*$ not occurring in $\Sigma$) and $V$ has exactly one $*$-labelled node, it follows that $U \cdot V$ also has exactly one $*$-labelled node; we conclude that $U \cdot V \in \PCsp$.
The other rules can be verified similarly.

For the other inclusion, let $U \in \PCsp$; we show that $U \in L$ by induction on the construction of $U$ as a series-parallel pomset.
In the base, the case where $U = 1$ can be discounted, for the empty pomset has no node to label with $*$.
We are thus left with the case where $U$ is primitive; the sole node must then be labelled with $*$, and hence $U = *$, meaning that $U \in L$.
For the inductive step, there are two cases to consider.
First, if $U = V \cdot W$ for series-parallel pomsets $V$ and $W$, then exactly one of these must contain exactly one $*$-labelled node --- in any other case, $U$ cannot occur in $\PCsp$.
Suppose that $V$ is this pomset, then $V \in \PCsp$, and by induction we find that $V \in L$; we also know that $W$ cannot contain any $*$-labelled node, and hence $W \in \SP$.
We can then conclude that $V \cdot W \in L$.
The case where $V$ contains no $*$-labelled node and $W$ has exactly one can be verified similarly.
Lastly, the case where $U = V \parallel W$ for series-parallel pomsets $V$ and $W$ can be treated analogously.
\end{proof}

To prove \autoref{lemma:context-prime-locus} and \autoref{lemma:context-preserve-series-parallel}, we need the following auxiliary lemma.

\begin{lemma}%
\label{lemma:substitution-vs-composition}
Let $C, C' \in \PC$ and $U, V \in \Pom$.
The following hold:
\begin{enumerate}[(i)]
    \item\label{property:substitution-vs-empty}
    If $C = *$, then $C[U] = U$.

    \item\label{property:substitution-vs-concat-right}
    If $C = C' \cdot V$, then $C[U] = C'[U] \cdot V$.

    \item\label{property:substitution-vs-concat-left}
    If $C = V \cdot C'$, then $C[U] = V \cdot C'[U]$.

    \item\label{property:substitution-vs-parallel}
    If $C = C' \parallel V$, then $C[U] = C'[U] \parallel V$.
\end{enumerate}
\end{lemma}
\begin{proof}
Let $C = [\lp{c}]$, $C' = [\lp{c}']$, $U = [\lp{u}]$ and $V = [\lp{v}]$.
For the first claim, we can assume without loss of generality that $S_\lp{c} = \{ s_* \}$ is disjoint from $S_\lp{u}$; it suffices to show that $\lp{c}[\lp{u}]$ is isomorphic to $\lp{u}$.
First, we note that $S_{\lp{c}[\lp{u}]} = S_\lp{u} \cup S_\lp{c} - \{ s_* \} = S_\lp{u}$.
Furthermore, if $s \in S_{\lp{c}[\lp{u}]} = S_\lp{u}$, then $\lambda_{\lp{c}[\lp{u}]}(s) = \lambda_\lp{u}(s)$ by definition.
Lastly, if $s, s' \in S_{\lp{c}[\lp{u}]}$ such that $s \leq_{\lp{c}[\lp{u}]} s'$, then since $s, s' \in S_\lp{u}$ we also have that $s \leq_\lp{u} s'$ by definition of $\leq_{\lp{c}[\lp{u}]}$; the other direction can be shown similarly.
Hence, we conclude that $C[U] = [\lp{c}[\lp{u}]] = [\lp{u}] = U$.

For the second claim, we can assume without loss of generality that $\lp{c}'$ is disjoint from $\lp{v}$, and that $\lp{c} = \lp{c}' \cdot \lp{v}$.
We can furthermore assume that $\lp{u}$ is disjoint from $\lp{c}$ (and hence from $\lp{c}'$ and $\lp{v}$, meaning in particular that $\lp{c}'[\lp{u}]$ is disjoint from $\lp{v}$).
It now suffices to show that $\lp{c}[\lp{u}]$ is the same as $\lp{c}'[\lp{u}] \cdot \lp{v}$.
First, we note that the carriers of these labelled posets are identical by construction:
\[
    S_{\lp{c}[\lp{u}]} = S_\lp{c} \cup S_\lp{u} - \{ s_* \} = S_{\lp{c}'} \cup S_\lp{v} \cup S_\lp{u} - \{ s_* \} = S_{\lp{c}'[\lp{u}]} \cup S_\lp{v} = S_{\lp{c}'[\lp{u}] \cdot \lp{v}}
\]
To see that their labellings coincide, suppose that $s \in S_{\lp{c}[\lp{u}]}$; there are three cases.
\begin{itemize}
    \item
    If $s \in S_{\lp{c}'} - \{ s_* \}$, then in particular $s \in S_\lp{c} - \{ s_* \}$, meaning
    \[
        \lambda_{\lp{c}[\lp{u}]}(s)
            = \lambda_\lp{c}(s)
            = \lambda_{\lp{c}'}(s)
            = \lambda_{\lp{c}'[\lp{u}]}(s)
            = \lambda_{\lp{c}'[\lp{u}] \cdot \lp{v}}(s)
    \]

    \item
    If $s \in S_\lp{v}$, then in particular $s \in S_\lp{c} - \{ s_* \}$, meaning
    \[
        \lambda_{\lp{c}[\lp{u}]}(s)
            = \lambda_{\lp{c}}(s)
            = \lambda_\lp{v}(s)
            = \lambda_{\lp{c}'[\lp{u}] \cdot \lp{v}}(s)
    \]

    \item
    If $s \in S_\lp{u}$, then we derive
    \(
        \lambda_{\lp{c}[\lp{u}]}(s)
            = \lambda_\lp{u}(s)
            = \lambda_{\lp{c}'[\lp{u}]}(s)
            = \lambda_{\lp{c}'[\lp{u}] \cdot \lp{v}}(s)
    \).
\end{itemize}

\noindent
To see that ${\leq_{\lp{c}[\lp{u}]}} = {\leq_{\lp{c}'[\lp{u}] \cdot \lp{v}}}$, first suppose that $s, s' \in S_{\lp{c}[\lp{u}]}$ with $s \leq_{\lp{c}[\lp{u}]} s'$.
\begin{itemize}
    \item
    If $s, s' \in S_\lp{c} - \{ s_* \}$, then $s \leq_\lp{c} s'$; this gives us three subcases to consider.
    \begin{itemize}
        \item
        If $s, s' \in S_{\lp{c}'} - \{ s_* \}$, then $s \leq_{\lp{c}'} s'$, meaning $s \leq_{\lp{c}'[\lp{u}]} s'$; thus, $s \leq_{\lp{c}'[\lp{u}] \cdot \lp{v}} s'$.

        \item
        If $s, s' \in S_\lp{v}$, then $s \leq_\lp{v} s'$, meaning that $s \leq_{\lp{c}'[\lp{u}] \cdot \lp{v}} s'$.

        \item
        If $s \in S_{\lp{c}'} - \{ s_* \}$ and $s' \in S_\lp{v}$, then $s \in S_{\lp{c}'[\lp{u}]}$, meaning $s \leq_{\lp{c}'[\lp{u}] \cdot \lp{v}} s'$.
    \end{itemize}

    \item
    If $s, s' \in S_\lp{u}$, then $s \leq_\lp{u} s'$. This tells us that $s \leq_{\lp{c}'[\lp{u}]} s'$, meaning $s \leq_{\lp{c}'[\lp{u}] \cdot \lp{v}} s'$.

    \item
    If $s \in S_\lp{u}$ and $s' \in S_\lp{c} - \{ s_* \}$ with $s_* \leq_{\lp{c}} s'$, then there are two subcases:
    \begin{itemize}
        \item
        If $s' \in S_{\lp{c}'} - \{ s_* \}$, then $s_* \leq_{\lp{c}'} s'$, meaning $s \leq_{\lp{c}'[\lp{u}]} s'$; thus, $s \leq_{\lp{c}'[\lp{u}] \cdot \lp{v}} s'$.

        \item
        If $s' \in S_\lp{v}$, then since $s \in S_{\lp{c'}[\lp{u}]}$, we have $s \leq_{\lp{c}'[\lp{u}] \cdot \lp{v}} s'$ immediately.
    \end{itemize}

    \item
    If $s \in S_\lp{c}$ and $s' \in S_\lp{u}$ with $s \leq_{\lp{c}} s_*$, an argument similar to the above applies.
\end{itemize}
This shows that ${\leq_{\lp{c}[\lp{u}]}} \subseteq {\leq_{\lp{c}'[\lp{u}] \cdot \lp{v}}}$; the other inclusion can be shown similarly.

The third and fourth claim can be proved using an argument analogous to the second claim.
\end{proof}

\contextprimelocus*
\begin{proof}
The proof proceeds by induction on the construction of $\PCsp$ as in \autoref{lemma:series-parallel-contexts-inductive}.
In the base, $C = *$.
We then have $U = C[U] = V \parallel W$.
Since $U$ is a non-empty word, we find that either $V = 1$ or $W = 1$ by~\cite[Lemma~3.1]{kappe-brunet-silva-zanasi-2018}.
In the former case, we choose $C' = C$ to find that $C'[U] = C[U] = U = W$, and $C = C \parallel 1 = C' \parallel V$; the latter case can be handled similarly.

For the inductive step, there are three cases to consider.
\begin{itemize}
    \item
    If $C = D \cdot X$ for some $D \in \PCsp$ and $X \in \SP$, then $D[U] \cdot X = V \parallel W$.
    By~\cite[Lemma~3.1]{kappe-brunet-silva-zanasi-2018}, there are three subcases to consider.
    \begin{itemize}
        \item
        If $V = 1$, then choose $C' = C$ s.t.\ $C'[U] = C[U] = W$ (by \autoref{lemma:substitution-vs-composition}\eqref{property:substitution-vs-concat-right}) and $C = C \parallel 1 = C' \parallel V$.
        The case where $W = 1$ is similar.

        \item
        If $D[U] = 1$, then $D = *$ and $U = 1$; but this contradicts the premise that $U$ is non-empty.
        We can thus exclude this case.

        \item
        If $X = 1$, then $C[U] = D[U] = V \parallel W$.
        The claim follows by induction.
    \end{itemize}

    \item
    If $C = X \cdot D$ for some $D \in \PCsp$ and $X \in \SP$, then we can find $C'$ in a manner analogous to the above.

    \item
    If $C = D \parallel X$ for some $D \in \PCsp$ and $X \in \SP$, then by~\cite[Lemma~3.5]{kappe-brunet-silva-zanasi-2018}, we obtain $Y_0, Y_1, Z_0, Z_1 \in \SP$ such that all of the following hold:
    \begin{mathpar}
    D[U] = Y_0 \parallel Y_1
    \and
    X = Z_0 \parallel Z_1
    \and
    V = Y_0 \parallel Z_0
    \and
    W = Y_1 \parallel Z_1
    \end{mathpar}
    By induction, we find $D' \in \PCsp$ such that either $D = D' \parallel Y_1$ and $D'[U] = Y_0$, or $D = Y_0 \parallel D'$ and $D'[U] = Y_1$.
    In the former case, we can choose $C' = D' \parallel Z_0$ to find that $C'[U] = D'[U] \parallel Z_0 = Y_0 \parallel Z_0 = V$ (by \autoref{lemma:substitution-vs-composition}\eqref{property:substitution-vs-parallel}) and $C = D \parallel X = D' \parallel Y_1 \parallel Z_0 \parallel Z_1 = C' \parallel W$.
    The latter case is similar.
    \qedhere
\end{itemize}
\end{proof}

\contextpreserveseriesparallel*
\begin{proof}
We proceed by induction on the inductive construction of $\PCsp$ given in \autoref{lemma:series-parallel-contexts-inductive}.
In the base, we have that $C = *$, and hence $C[U] = U$ by \autoref{lemma:substitution-vs-composition}\eqref{property:substitution-vs-empty}; since $U \in \SP$, the claim follows.
For the inductive step, first suppose that $C = V \cdot C'$ for some $C' \in \PCsp$ and $V \in \SP$
By \autoref{lemma:substitution-vs-composition}\eqref{property:substitution-vs-concat-left}, we have that $C[U] = V \cdot C'[U]$; since $C'[U] \in \SP$ by induction, it then follows that $C[U] \in \SP$ as well.
The other cases can be treated similarly.
\end{proof}

To prove \autoref{lemma:context-erasure}, we first prove two auxiliary lemmas.
First, we describe the labelled posets involved when the empty pomset is plugged into a context.

\begin{lemma}%
\label{lemma:context-plug-empty}
Let $C \in \PC$ and $U \in \Pom$.
Now $C[1] = U$ if and only if $C = [\lp{c}]$ and $U = [\lp{u}]$ such that:
\begin{enumerate}[(i)]
    \item
    $S_\lp{c} = S_\lp{u} \cup \{ s_* \}$ with $s_* \not\in S_\lp{u}$, and

    \item
    for all $s, s' \in S_\lp{u}$ it holds that $s \leq_\lp{u} s'$ if and only if $s \leq_\lp{c} s'$, and

    \item
    $\lambda_\lp{c}(s_*) = *$ and $\lambda_\lp{c}(s) = \lambda_\lp{u}(s)$ for all $s \in S_\lp{u}$.
\end{enumerate}
\end{lemma}
\begin{proof}
Let $C = [\lp{c}]$, where $s_* \in S_\lp{c}$ is the unique $*$-labelled note of $\lp{c}$, and write $\lp{1}$ for the (unique) empty labelled partial order.
For the direction from left to right, we know that $U = C[1] = [\lp{c}[\lp{1}]]$.
It suffices to prove that $\lp{c}[\lp{1}]$ satisfies exactly the properties of $\lp{u}$ listed above.
First of all, note that $S_{\lp{c}[\lp{1}]} = S_\lp{c} \cup S_\lp{1} - \{ s_* \}$ by definition, hence $s_* \not\in S_{\lp{c}[\lp{1}]}$.
Furthermore, since $S_\lp{1} = \emptyset$, it follows that $S_{\lp{c}[\lp{1}]} \cup \{ s_* \} = S_\lp{c}$.
Next, suppose that $s, s' \in S_{\lp{c}[\lp{1}]}$ with $s \leq_{\lp{c}[\lp{1}]} s'$.
We can discount the possibility that $s \in S_\lp{1}$ or $s' \in S_\lp{1}$, which leaves us to conclude that $s \leq_\lp{c} s'$; the converse holds by definition of $\leq_{\lp{c}[\lp{1}]}$.
Lastly, note that $\lambda_\lp{c}(s_*) = *$ immediately, and that if $s \in S_{\lp{c}[\lp{1}]}$, then $s \in S_\lp{c}$, and hence $\lambda_{\lp{c}[\lp{1}]}(s) = \lambda_\lp{c}(s)$.

Conversely, we can show that $\lp{c}[\lp{1}]$ is isomorphic to $\lp{u}$ satisfying the above conditions, and hence that $C[1] = U$.
In detail, first note that $S_{\lp{c}[\lp{1}]} = S_\lp{c} \cup S_\lp{1} - \{ s_* \} = S_\lp{c} - \{ s_* \} = S_\lp{u}$ by the first property; we choose the identity on $S_\lp{c}$ to be the mediating isomorphism.
To see that this indeed gives us a labelled poset isomorphism, note that the identity preserves and reflects ordering by the first property, and it preserves labels by the second property.
\end{proof}

The second auxiliary lemma that we will need has to do with the second part of the claim for \autoref{lemma:context-erasure}.
It says that we can turn any context into a series-parallel context while preserving the effect of plugging in the empty pomset; in fact this new context will be subsumed by the old one.

\begin{lemma}%
\label{lemma:context-to-sp-context}
Let $C\in\PC$ and $U \in \SP$ such that $C[1] = U$.
We can construct a $C' \in \PCsp$ such that $C'[1] = U$ and $C'\sqsubseteq C$.
\end{lemma}
\begin{proof}
Let $C=[\lp{c}]$ and $U=[\lp{u}]$.
We will show that if $C\notin\PCsp$, then we can construct a $\lp{c}'$ such that
\begin{enumerate}[(i)]
    \item
    $S_\lp{c} =S_{\lp{c}'}$, and
    \item
    for all $s, s' \in S_\lp{c} \setminus \{s_*\}$ with $s \leq_{\lp{c}'} s'$, we have $s \leq_\lp{c} s'$, and
    \item
    $\lambda_\lp{c} = \lambda_{\lp{c}'}$, and
    \item
    ${\leq_\lp{c}}$ is contained in but not equal to ${\leq_{\lp{c}'}}$.
\end{enumerate}
The first three conditions, in combination with \autoref{lemma:context-plug-empty}, imply that $C[1] = [\lp{c}'[\lp{1}]]$.
Moreover, the first, second and last condition together tell us that $[\lp{c}'] \sqsubseteq [\lp{c}]$ but $[\lp{c}'] \neq [\lp{c}]$.
Hence, $[\lp{c}']$ is strictly smaller than $C$ but still satisfies the premise of the lemma.
By well-founded induction on $\sqsubseteq$, we can conclude that if we repeat this process we eventually end up with a context $C'\in\PCsp$ that is subsumed by $C$ and still satisfies the condition that $C'[1] = U$.

An \N-pattern in $\lp{v}$ is a quadruplet of events $(s_1,s_2,s_3,s_4)\in S_\lp{v}^4$ such that:
\begin{mathpar}
s_1\leq_\lp{v} s_3
\and
s_2\leq_\lp{v} s_3
\and
s_2\leq_\lp{v} s_4
\and
s_1\not\leq_\lp{v} s_4
\and
s_2\not\leq_\lp{v} s_1
\and
s_4\not\leq_\lp{v} s_3.
\end{mathpar}
It is well known that a pomset is series-parallel if and only if it does not contain any \N-pattern~\cite{gischer-1988}.
Since $C \not\in \PCsp$ but $C \in \PC$, it follows that $C$ is not series-parallel; hence, there is an \N-pattern $(s_1,s_2,s_3,s_4) \in S_\lp{c}^4$.
On the other hand, since $U = C[1] \in \SP$ we know that $U$ does not have this pattern, so $s_*$ must be one of these four events.
Let us do a case analysis:
\begin{enumerate}
    \item
    First, suppose that $s_* = s_1$.
    We claim that we can build $\lp{c}'$ by choosing
    \begin{mathpar}
    S_{\lp{c}'} = S_\lp{c}
    \and
    \lambda_{\lp{c}'} = \lambda_\lp{c}
    \and
    {\leq_{\lp{c}'}} = {\left({\leq_\lp{c}} \cup\{(s_*,s_4)\}\right)}^*
    \end{mathpar}
    Clearly, the first, second and last conditions on $\lp{c}'$ hold by construction, and $\leq_{\lp{c}'}$ is reflexive and transitive, too.
    It remains to validate the third condition, and that $\leq_\lp{c'}$ is antisymmetric.
    The following facts help establish this.

    \begin{fact}%
    \label{fact:sstar-vs-s4}
    For all $s \in S_\lp{c}$ with $s <_\lp{c} s_*$ we have $s \leq_\lp{c} s_4$.
    \end{fact}
    \begin{proof}[Proof of \autoref{fact:sstar-vs-s4}]
    The proof proceeds by contradiction: assume there exists $s<_\lp{c} s_*$ with $s\not\leq_\lp{c} s_4$.
    Then we can show that the quadruplet $(s,s_2,s_3,s_4) \in S_\lp{c} \setminus\{s_*\} = S_\lp{u}$ is an \N-pattern in $\lp{u}$.
    Indeed, we already know that
    \begin{mathpar}
    s_2\leq_\lp{u}s_3
    \and
    s_2\leq_\lp{u}s_4
    \and
    s_4\not\leq_\lp{u}s_3.
    \end{mathpar}
    Therefore what remains are the statements relating to $s$, i.e., that
    \begin{mathpar}
    s\leq_\lp{u}s_3
    \and
    s\not\leq_\lp{u}s_4
    \and
    s_2\not\leq_\lp{u}s
    \end{mathpar}
    The first one is obtained by transitivity: $s\leq_\lp{c} s_*\leq_\lp{c}s_3$, and thus $s \leq_\lp{u} s_3$.
    The second one is a direct consequence of our assumption that $s \not\leq_\lp{c} s_4$.
    Lastly, if we assume $s_2 \leq_\lp{u} s$, then $s_2 \leq_\lp{c} s$, and by transitivity we get $s_2 \leq_\lp{c} s_*$, which contradicts that $(s_*,s_2,s_3,s_4)$ is an \N-shape.

    \smallskip
    We now have shown that $s \not\leq_\lp{c} s_4$ implies the existence of an \N-pattern in $U$, which cannot be the case.
    We conclude that if $s <_\lp{c}  s_*$, then $s \leq_\lp{c} s_4$.
    \end{proof}

    \begin{fact}%
    \label{fact:cprime-vs-c}
    If $s \leq_{\lp{c}'} s'$, then either $s \leq_\lp{c} s'$, or $s = s_*$ and $s_4 \leq_\lp{c} s'$.
    \end{fact}
    \begin{proof}[Proof of \autoref{fact:cprime-vs-c}]
    We shall phrase the claim in terms of relational algebra, and prove it by reasoning on relations.
    Let $R = \{(s_*,s_4)\}$ and $T = {<_\lp{c}}$; note that ${\leq_\lp{c}} = {(T \cup R)}^*$.
    The claim is now equivalent to showing that
    \[
        {(T \cup R)}^* \subseteq T^* \cup R \circ T^*
    \]

    \smallskip
    To this end, note that \autoref{fact:sstar-vs-s4} can be written as $T \circ R\subseteq T^*$.
    Also, $R \circ R = \emptyset$, since $s_*\neq s_4$, because $s_* \leq_\lp{c} s_3$ and $s_4 \not\leq_\lp{u} s_3$.
    Using these properties, we can derive the following containments:
    \begin{align*}
        T \circ T^*
            &\subseteq T^*\cup R \circ T^*\\
        T \circ R \circ T^*
             \subseteq T^*\circ T^*
            &\subseteq T^*\cup R\circ T^*\\
        R \circ T^*
            &\subseteq T^*\cup R\circ T^*\\
        R \circ R \circ T^*
             = \emptyset \circ T^*
             = \emptyset
            &\subseteq T^*\cup R\circ T^*
    \end{align*}
    By distributivity, we can then derive the following:
    \begin{align*}
        \left(T\cup R\right)\circ\left(T^*\cup R\circ T^*\right)
        &=
            T\circ T^* \cup
            T\circ R\circ T^*\cup
            R\circ T^* \cup
            R\circ R\circ T^*\\
        &\subseteq T^*\cup R\circ T^*.
    \end{align*}
    By the fixpoint principle for reflexive-transitive closure, it follows that:
    \[
        {\left(T\cup R\right)}^\star \circ \left( T^*\cup R\circ T^* \right)
            \subseteq T^* \cup R \circ T^*.
    \]
    Finally, we conclude that ${\left( T \cup R \right)}^* \subseteq T^* \cup R\circ T^*$ by:
    \[
        {\left(T\cup R\right)}^*
            ={\left(T\cup R\right)}^* \circ \mathsf{id}
            \subseteq {\left( T \cup R \right)}^* \circ \left( T^*\cup R\circ T^* \right)
            \subseteq T^*\cup R\circ T^*
        \qedhere
    \]
    \end{proof}

    We can now use \autoref{fact:cprime-vs-c} to show the remaining properties.
    For the third condition on $\lp{c}'$, assume $s \leq_\lp{c'} s'$ with $s,s' \neq s_*$.
    By the previous observation, either $s \leq_\lp c s'$ or we have both $s = s_*$ and $s_4\leq_c s'$.
    Since $s \neq s_*$, we may conclude that $s \leq_\lp c s'$.

    \smallskip
    For antisymmetry, let $s \leq_\lp{c'} s' \leq_\lp{c'} s$.
    Using \autoref{fact:cprime-vs-c}, we distinguish four cases:
    \begin{enumerate}
        \item
        If $s \leq_\lp{c} s' \leq_\lp{c} s$, then by antisymmetry of $\leq_\lp{c}$ we get $s = s'$;

        \item
        If $s \leq_\lp{c} s'$ and $s' = s_*$ with $s_4 \leq_\lp{c} s$, then we get $s_4 \leq_\lp{c} s \leq_\lp{c} s' = s_*$;

        \item
        If $s = s_*$ with $s_4 \leq_\lp{c} s'$ and $s' \leq_\lp{c} s$, then we get $s_4 \leq_\lp{c} s' \leq_\lp{c} s = s_*$;

        \item
        If $s = s_*$ with $s_4 \leq_\lp{c} s'$ and $s' = s_*$ with $s_4 \leq_\lp{c} s$, then we get $s = s_* = s'$.
    \end{enumerate}
    In the first and last case we could conclude that $s = s'$, while in the other three we ended up with $s_4\leq_\lp c s_*$, contradicting that $s_* \leq_\lp{c} s_3$ and $s_4 \not\leq_\lp{u} s_3$.

    \item
    Next, suppose that $s_*=s_2$.
    We claim that we can build $\lp{c}'$ by choosing
    \begin{mathpar}
    S_{\lp{c}'} = S_\lp{c}
    \and
    \lambda_{\lp{c}'} = \lambda_\lp{c}
    \and
    {\leq_{\lp{c}'}} = {\left({\leq_\lp{c}} \cup\{(s_*,s_1)\}\right)}^*
    \end{mathpar}
    As before, the first, second and last conditions on $\lp{c}'$ hold by construction, and $\leq_{\lp{c}'}$ is reflexive and transitive, too.
    It remains to validate the third condition, and that $\leq_\lp{c'}$ is antisymmetric.
    The following facts help establish this.

    \begin{fact}%
    \label{fact:sstar-vs-s1}
    For all $s \in S_\lp{c}$ with $s <_\lp{c} s_*$, we have $s \leq_\lp{c} s_1$.
    \end{fact}
    \begin{proof}[Proof of \autoref{fact:sstar-vs-s1}]
    We proceed by contradiction.
    Assume there exists $s <_\lp{c} s_*$ such that $s\not\leq_\lp{c} s_1$
    Then we can show that the quadruplet $(s_1,s,s_3,s_4)\in S_\lp{c} \setminus\{s_*\}=S_\lp{u}$ is an \N-pattern.
    Indeed, we already know that
    \begin{mathpar}
    s_1 \leq_\lp{u} s_3
    \and
    s_1 \not\leq_\lp{u} s_4
    \and
    s_4 \not\leq_\lp{u} s_3.
    \end{mathpar}
    Therefore what remains are the statements relating to $s$, i.e.,
    \begin{mathpar}
    s \leq_\lp{u} s_3
    \and
    s \leq_\lp{u} s_4
    \and
    s \not\leq_\lp{u} s_1
    \end{mathpar}
    The first and second are obtained by transitivity: $s \leq_\lp{c} s_*\leq_\lp{c} s_3,s_4$.
    The third one follows directly from our assumption.

    \smallskip
    We now have shown that    $s\not\leq_\lp{c} s_1$ implies the existence of an \N-pattern in $U$, which cannot be the case.
    We conclude that if $s <_\lp{c} s_*$, then $s \leq_\lp{c} s_1$.
    \end{proof}

    \begin{fact}%
    \label{fact:cprime-vs-c-bis}
    If $s \leq_{\lp{c}'} s'$, then either $s \leq_\lp{c} s'$ or $s = s_*$ and $s_1 \leq_\lp{c} s'$.
    \end{fact}
    \begin{proof}[Proof of \autoref{fact:cprime-vs-c-bis}]
    Just like in the proof of \autoref{fact:cprime-vs-c}, we can formulate the claim in terms of relational algebra.
    Let $R = \{(s_*, s_1)\}$ and $T = {<_\lp{c}}$.
    As before, we have ${\leq_\lp{c}} = {(T \cup R)}^*$, which makes the claim equivalent to showing that
    \[
        {\left( T \cup R \right)}^*
            \subseteq T^* \cup R \circ T^\star
    \]
    We can reformulate \autoref{fact:sstar-vs-s1} as $T \circ R \subseteq T^*$.
    Note also that $T \circ T = \emptyset$, since $s_*\neq s_1$, because $s_*\leq_\lp{c}s_4$ and $s_1\not\leq_\lp{u}s_4$.
    Since we now have the same hypotheses as in \autoref{fact:cprime-vs-c}, we may derive the same conclusion.
    \end{proof}

    The desired properties now follow from \autoref{fact:cprime-vs-c-bis}, as in the previous case.

    \item
    If $s_* = s_3$, then by a similar argument as in the case where $s_* = s_2$, we may show first that for all $s$ with $s_* <_\lp{c} s$ it holds that $s_4 <_\lp{c} s$.
    We can then use this to show that choosing ${\leq_{\lp{c}'}} = {({\leq_\lp{c}} \cup \{(s_4,s_*)\})}^*$ validates the claim.

    \item
    If $s_* = s_4$, then by a similar argument as in the case where $s_* = s_1$, we may show first that for all $s$ with $s_* <_\lp{c} s$ it holds that $s_1 <_\lp{c} s$.
    We can then use this to show that choosing ${\leq_{\lp{c}'}} = {({\leq_\lp{c}} \cup \{(s_1,s_*)\})}^*$ validates the claim.
    \qedhere
\end{enumerate}
\end{proof}

\noindent
With these lemmas in place, we can now prove \autoref{lemma:context-erasure}.

\contexterasure*
\begin{proof}
Let $U = C[1]$.
It suffices to construct a $C' \in \PC$ such that $C' \sqsubseteq C$ and $C'[1] = V$, since \autoref{lemma:context-to-sp-context} takes care of the ``moreover'' clause.
By \autoref{lemma:context-plug-empty}, we find that $C = [\lp{c}]$ and $U = [\lp{u}]$ such that $S_\lp{c} = S_\lp{u} \cup \{ * \}$, for all $s, s' \in S_\lp{u}$ it holds that $s \leq_\lp{u} s'$ if and only if $s \leq_\lp{c} s'$, $\lambda_\lp{c}(*) = *$ and $\lambda_\lp{c}(s) = \lambda_\lp{u}(s)$ for all $s \in S_\lp{u}$.
Let $V = [\lp{v}]$; since $V \sqsubseteq U$, we know without loss of generality that $S_\lp{v} = S_\lp{u}$ and $\lambda_\lp{u} = \lambda_\lp{v}$ and ${\leq_\lp{u}} \subseteq {\leq_\lp{v}}$.

Let $\leq_{\lp{c}'}$ be the smallest transitive relation on $S_\lp{c}$ containing both $\leq_{\lp{v}}$ and $\leq_\lp{c}$.
Let $s, s' \in S_\lp{c}$ with $s \leq_{\lp{c}'} s'$; the following properties of $\leq_{\lp{c}'}$ will be useful:
\begin{enumerate}[(i)]
    \item\label{property:interpolate-order-left}
    If $s = s_*$ and $s' \in S_\lp{v}$, then there exists an $\hat{s} \in S_\lp{v}$ such that $s_* \leq_\lp{c} \hat{s} \leq_\lp{v} s'$.

    \item\label{property:interpolate-order-right}
    If $s \in S_\lp{v}$ and $s' = s_*$, then there exists an $\hat{s} \in S_\lp{v}$ such that $s \leq_\lp{v} \hat{s} \leq_\lp{c} s_*$.

    \item\label{property:specialise-order}
    If $s, s' \in S_\lp{v}$, then $s \leq_\lp{v} s'$.
\end{enumerate}
We prove these claims in tandem by induction on the construction of $\leq_{\lp{c}'}$.
In the base, suppose for the first claim that $s = s_*$ and $s' \in S_\lp{v}$; we then know that $s \leq_{\lp{c}} s'$ (the case where $s \leq_\lp{v} s'$ can be excluded, for $s_* \not\in S_\lp{v}$), and hence we can choose $\hat{s} = s'$ to satisfy the claim; the second claim goes through similarly.
For the last claim, if $s \leq_{\lp{v}} s'$ then we are done immediately; otherwise, if $s \leq_\lp{c} s'$, then since $s, s' \neq s_*$ we have that $s \leq_\lp{u} s'$, and hence $s \leq_\lp{v} s'$.

In the inductive step, we have that $s \leq_{\lp{c}'} s'$ because there exists an $s'' \in S_\lp{c}$ with $s \leq_{\lp{c}'} s'' \leq_{\lp{c}'} s'$.
We consider each claim separately.
\begin{enumerate}[(i)]
    \item
    If $s = s_*$ and $s' \in S_\lp{v}$, then there are two cases to consider.
    \begin{itemize}
        \item
        If $s'' = s_*$, then we can apply the induction hypothesis to $s'' \leq_{\lp{c}'} s'$ to find an $\hat{s} \in S_\lp{v}$ such that $s_* \leq_\lp{c} \hat{s} \leq_\lp{v} s'$.

        \item
        If $s'' \neq s_*$, then $s'' \in S_\lp{v}$, and we can apply the induction hypothesis to $s \leq_{\lp{c}'} s''$ to find an $\hat{s} \in S_\lp{v}$ such that $s_* \leq_\lp{c} \hat{s} \leq_\lp{v} s''$.
        By applying the induction hypothesis (specifically, the third claim) to $s'' \leq_{\lp{c}'} s'$, we find that $s'' \leq_\lp{v} s'$, and thus we can conclude that $s_* \leq_\lp{c} \hat{s} \leq_\lp{v} s'$.
    \end{itemize}

    \item
    If $s \in S_\lp{v}$ and $s = s_*$, then the proof proceeds as in the previous case.

    \item
    If $s, s' \in S_\lp{v}$, then there are again two cases to consider.
    \begin{itemize}
        \item
        If $s'' = s_*$, then we can apply the induction hypothesis (specifically, the second claim) to $s \leq_{\lp{c}'} s''$ to find an $\hat{s} \in S_\lp{v}$ such that $s \leq_\lp{v} \hat{s} \leq_\lp{c} s''$.
        Similarly, we can apply the induction hypothesis (in this case, the first claim) to $s'' \leq_{\lp{c}'} s$ to find an $\hat{s}' \in S_\lp{v}$ such that $s'' \leq_\lp{c} \hat{s}' \leq_\lp{v} s'$.
        We then know that $\hat{s} \leq_\lp{c} \hat{s}'$, and since $\hat{s}, \hat{s}' \in S_\lp{v} = S_\lp{u}$, we know that $\hat{s} \leq_\lp{u} \hat{s}'$, and hence $\hat{s} \leq_\lp{v} \hat{s}'$.
        In total, we find that $s \leq_\lp{v} \hat{s} \leq_\lp{v} \hat{s}' \leq_\lp{v} s'$.

        \item
        If $s'' \neq s_*$, then $s'' \in S_\lp{v}$, and we can apply the induction hypothesis (specifically, the third claim) to both $s \leq_{\lp{c}'} s''$ and $s'' \leq_{\lp{c}'} s'$ to find that $s \leq_\lp{v} s'' \leq_\lp{v} s'$.
    \end{itemize}
\end{enumerate}

\noindent
We now claim that $\leq_{\lp{c}'}$ is antisymmetric.
To see this, suppose that $s, s' \in S_\lp{c}$ with $s \leq_{\lp{c}'} s' \leq_{\lp{c}'} s$.
Now, if $s, s' \in S_\lp{v}$, then $s \leq_\lp{v} s' \leq_\lp{v} s$ by property~\eqref{property:specialise-order}, and hence $s = s'$ by antisymmetry of $\leq_\lp{v}$.
Otherwise, if $s = s_*$, then suppose towards a contradiction that $s' \neq s_*$; in that case, $s' \in S_\lp{v}$, and we can find $\hat{s}, \hat{s}' \in S_\lp{v}$ such that $s_* \leq_\lp{c} \hat{s} \leq_\lp{v} s' \leq_\lp{v} \hat{s}' \leq_\lp{c} s_*$ by properties~\eqref{property:interpolate-order-left} and~\eqref{property:interpolate-order-right}.
But then, since $\hat{s}' \leq_\lp{c} \hat{s}$, it follows that $\hat{s}' \leq_\lp{v} \hat{s}$.
Moreover, $\hat{s} \leq_\lp{v} s' \leq_\lp{v} \hat{s}' \leq \hat{s}$, and hence $\hat{s}' = s' = \hat{s}$ by antisymmetry of $\leq_\lp{v}$.
It then follows that $s_* \leq_\lp{c} s' \leq_\lp{c} s_*$, meaning that $s' = s_*$ by antisymmetry of $\leq_\lp{c}$ --- a contradiction.
We conclude that $s' = s_* = s$.

Since $\leq_{\lp{c}'}$ is reflexive and transitive by construction, and antisymmetric by the above, it is a partial order.
We now choose $S_{\lp{c}'} = S_\lp{c}$ and $\lambda_{\lp{c}'} = \lambda_{\lp{c}}$, and let $C' = [\lp{c}']$.
Note that $C'$ has exactly one $s_*$-labelled node, and hence $C' \in \PC$.
Now, if $s, s' \in S_\lp{v}$, then $s \leq_\lp{v} s'$ implies $s \leq_{\lp{c}'} s'$ by definition of $\leq_{\lp{c}'}$; furthermore, if $s \leq_{\lp{c}'} s'$, then $s \leq_\lp{v} s'$ by property~\eqref{property:specialise-order} above.
Since $S_{\lp{c}'} = S_\lp{c} = S_\lp{u} \cup \{ s_* \} = S_\lp{v} \cup \{ s_* \}$, and furthermore
$\lambda_{\lp{c}'}(s_*) = \lambda_\lp{c}(s_*) = *$ and $\lambda_{\lp{c}'}(s) = \lambda_\lp{c}(s) = \lambda_\lp{u}(s) = \lambda_\lp{v}(s)$ for $s \in S_\lp{v}$,
we have that $C'[1] = V$ by \autoref{lemma:context-plug-empty}. Lastly, ${\leq_{\lp{c}}} \subseteq {\leq_{\lp{c}'}}$ by construction, and thus we have that $C' \sqsubseteq C$.
\end{proof}

To prove \autoref{lemma:context-subsumption}, we need the following auxiliary lemma, which is analogous to \autoref{lemma:context-plug-empty} except that it concerns plugging in single letters.

\begin{lemma}%
\label{lemma:context-plug-letter}
Let $C \in \PC$ and $U \in \Pom$ and $a \in \Sigma$.
Now $C[\ltr{a}] = U$ if and only if $C = [\lp{c}]$ and $U = [\lp{u}]$ such that the following hold:
\begin{enumerate}[(i)]
    \item
    $S_\lp{c} = S_\lp{u}$ as well as ${\leq_\lp{c}} = {\leq_\lp{u}}$, and

    \item
    $\lambda_\lp{c}(s_*) = *$ and $\lambda_\lp{u}(s_*) = \ltr{a}$, and

    \item
    $\lambda_\lp{c}(s) = \lambda_\lp{u}(s)$ for all $s \in S_\lp{c} - \{ s_* \}$.
\end{enumerate}
\end{lemma}
\begin{proof}
Let $C = [\lp{c}]$ with $s_* \in S_\lp{c}$ the unique node of $\lp{c}$ such that $\lambda_\lp{c}(s_*) = *$.
Also, let $\ltr{a} = [\lp{a}]$, where we assume without loss of generality that $S_\lp{a} = \{ s_\ltr{a} \}$ (where $s_\ltr{a}$ does not occur in $S_\lp{c}$), and we know that $\lambda_{\lp{a}} = s_\ltr{a} \mapsto \ltr{a}$.

For the direction from left to right, we choose $S_\lp{u} = S_\lp{c}$ and ${\leq_\lp{u}} = {\leq_\lp{c}}$, and we set $\lambda_\lp{u}(s_*) = \ltr{a}$, while $\lambda_\lp{u}(s) = \lambda_\lp{c}(s)$ for all $s \in S_\lp{c} - \{ s_* \}$.
It should be clear that this choice of $\lp{u}$ and $\lp{c}$ satisfies the three conditions above; it remains to prove that $[\lp{u}] = U$, for which it suffices to show that that $\lp{u}$ is isomorphic to $\lp{c}[\lp{a}]$, since $C[\ltr{a}] = U$.
To see this, first note that
\[
    S_{\lp{c}[\lp{a}]}
        = S_\lp{c} \cup S_\lp{a} - \{ s_* \}
        = S_\lp{c} \cup \{ s_\ltr{a} \} - \{ s_* \}
\]
We choose $h: S_{\lp{c}[\lp{a}]} \to S_\lp{u}$ by setting $h(s_\ltr{a}) = s_*$ and $h(s) = s$ when $s \neq s_\ltr{a}$; clearly, $h$ is a bijection between $S_{\lp{c}[\lp{a}]}$ and $S_\lp{u}$.
To see that $h$ preserves labels, first note
\[
    \lambda_\lp{u}(h(s_\ltr{a}))
        = \lambda_\lp{u}(s_*)
        = \ltr{a}
        = \lambda_\lp{a}(s_\ltr{a})
        = \lambda_{\lp{c}[\lp{a}]}(s_\ltr{a})
\]
Second, when $s \neq s_\ltr{a}$ we have that $s \in S_\lp{c}$, and hence
\[
    \lambda_\lp{u}(h(s))
        = \lambda_\lp{u}(s)
        = \lambda_\lp{c}(s)
        = \lambda_{\lp{c}[\lp{a}]}(s)
\]

To see that $h$ preserves order, suppose that $s, s' \in S_{\lp{c}[\lp{a}]}$ such that $s \leq_{\lp{c}[\lp{a}]} s'$; there are four cases to consider.
\begin{itemize}
    \item
    If $s \leq_\lp{c} s'$, then $s, s' \neq s_\ltr{a}$, and hence $h(s) = s \leq_\lp{u} s' = h(s')$ by definition.

    \item
    If $s \leq_\lp{a} s'$, then $s = s' = s_\ltr{a}$ since $S_\lp{a}$ is a singleton; hence $h(s_\ltr{a}) \leq_\lp{u} h(s_\ltr{a})$.

    \item
    If $s \leq_\lp{c} s_*$ and $s' \in S_\lp{a}$, then $s \in S_\lp{c}$ (and thus $s \neq s_\ltr{a}$) and $s' = s_\ltr{a}$; hence $h(s) = s \leq_\lp{c} s_* = h(s')$, meaning $h(s) \leq_\lp{u} h(s')$.

    \item
    If $s_* \leq_\lp{c} s'$ and $s \in S_\lp{a}$, then $s' \in S_\lp{c}$ (and thus $s' \neq s_\ltr{a}$) and $s = s_\ltr{a}$; hence $h(s) = s_* \leq_\lp{c} s' = h(s')$, meaning $h(s) \leq_\lp{u} h(s')$.
\end{itemize}
A similar argument shows that $h$ reflects ordering; hence, $h$ is a pomset isomorphism, and thus $U = [\lp{c}[\lp{a}]] = [\lp{u}]$.

For the converse direction, suppose that $U = [\lp{u}]$ such that the three conditions above are satisfied.
It remains to show that $C[\ltr{a}] = U$ --- in other words, that $\lp{c}[\lp{a}]$ is isomorphic to $\lp{u}$.
As isomorphism, we choose the identity function, which is already a bijection by the first property; it also preserves and reflects ordering (by the first property), and preserves labels (by the second and last properties).
Hence, $\lp{c}[\lp{a}]$ is isomorphic to $\lp{u}$, and hence $C[\ltr{a}] = U$.
\end{proof}

With this lemma in hand, we are now ready to prove \autoref{lemma:context-subsumption}.

\contextsubsumption*
\begin{proof}
Let $U = C[\ltr{a}]$.
By \autoref{lemma:context-plug-letter}, we know that $C = [\lp{c}]$ and $U = [\lp{u}]$ such that $S_\lp{u} = S_\lp{c}$ and ${\leq_\lp{u}} = {\leq_\lp{c}}$, with $\lambda_\lp{c}(s_*) = *$ and $\lambda_\lp{u}(s_*) = \ltr{a}$, and that $\lambda_\lp{c}(s) = \lambda_\lp{u}(s)$ for all $s \in S_\lp{c} - \{ s_* \}$.
Without loss of generality, we can assume that $V = [\lp{v}]$ with $S_\lp{v} = S_\lp{u}$ and $\lambda_\lp{v} = \lambda_\lp{u}$ and ${\leq_\lp{u}} \subseteq {\leq_\lp{v}}$.
We now choose $S_{\lp{c}'} = S_\lp{c}$ and ${\leq_{\lp{c}'}} = {\leq_\lp{v}}$ and $\lambda_{\lp{c}'} = \lambda_\lp{c}$ to obtain $C' = [\lp{c}']$.
First, note that $C' \sqsubseteq C$ by construction.
Also, observe that $S_{\lp{c}'} = S_\lp{c} = S_\lp{u} = S_\lp{v}$; furthermore, $\lambda_{\lp{c}'}(s_*) = \lambda_\lp{c}(s_*) = *$, while $\lambda_\lp{v}(s_*) = \lambda_\lp{u}(s_*) = \ltr{a}$, and for all $s \in S_\lp{c} - \{ * \}$ we have $\lambda_{\lp{c}'}(s) = \lambda_\lp{c}(s) = \lambda_\lp{u}(s) = \lambda_\lp{v}(s)$.
By \autoref{lemma:context-plug-letter}, we conclude that $C'[\ltr{a}] = V$.

Finally, if $V \in \SP$, then $C'$ must also be $\mathsf{N}$-free (and hence series-parallel), since any $\mathsf{N}$ in $C'$ must also occur in $V$ (by construction of $C'$), and $V$ is $\mathsf{N}$-free because it is series-parallel.
\end{proof}

\compositionversusclosure*
\begin{proof}
We treat the claims in the order given.
\begin{enumerate}
    \item
    We show both implications:
    \begin{itemize}
        \item
        First, suppose that $L \subseteq {\closure K}$.
        We show by induction that
        \[
                \forall A\subseteq \closure L,\,A\subseteq \closure K.
        \]
        \begin{itemize}
            \item
            In the base, where $A = L$, we have $A = L \subseteq {\closure K}$ by the premise.

            \item
            If $A=C[\semcka{e}]$ with $e \leq f \in H$ and $C[\semcka{f}] \subseteq \closure{L}$, then by the induction hypothesis $C[\semcka f]\subseteq\closure K$, and thus $A = C[\semcka{e}] \subseteq \closure{K}$.
        \end{itemize}

        \item
        The other implication is trivial: if ${\closure L} \subseteq {\closure K}$, then since $L\subseteq\closure L$ we get $L\subseteq\closure K$ by transitivity.
    \end{itemize}

    \item
    By definition of closure we have $K \subseteq {\closure K}$, and thus $L \subseteq K \subseteq {\closure K}$.
    By~\eqref{property:hypothesis-closure} we then immediately obtain the desired result.

    \item
    For the inclusion from left to right, first note that $L \cup K \subseteq {\closure L} \cup {\closure K}$ by definition of closure, and hence $\closurep {L \cup K} \subseteq \closurep {{\closure L} \cup {\closure K}}$ by~\eqref{property:hypothesis-monotone}.

    \smallskip
    For the other inclusion, note that ${\closure L}, {\closure K} \subseteq \closurep {L \cup K}$ by~\eqref{property:hypothesis-monotone}, and hence ${\closure L} \cup {\closure K} \subseteq \closurep {L \cup K}$.
    We then conclude by~\eqref{property:hypothesis-closure} that
    \[
            \closurep {{\closure L} \cup {\closure K}} \subseteq \closurep {L \cup K}
    \]

    \item
    For the inclusion from left to right, note that $L \subseteq {\closure L}$ and $K \subseteq {\closure K}$ by~\eqref{property:hypothesis-closure}.
    We then know that $L \cdot K \subseteq {\closure L} \cdot {\closure K}$; the inclusion then follows by~\eqref{property:hypothesis-monotone}.

    \smallskip
    For the other inclusion, it suffices to show that ${\closure L} \cdot {\closure K} \subseteq \closurep {L \cdot K}$ by~\eqref{property:hypothesis-closure}.
    We show by induction on $A\subseteq \closure L$ the following statement:
    \[
        \forall A \subseteq \closure L, A \cdot \closure K \subseteq \closurep{L\cdot K}.
    \]
    \begin{itemize}
        \item
        If $A = L$, then we do an induction on $B\subseteq \closure K$ to prove:
        \[
            \forall B \subseteq \closure K, L \cdot B \subseteq \closurep{L\cdot K}.
        \]
        \begin{itemize}
            \item
            In the base, where $B=K$, the statement holds by~\eqref{property:hypothesis-closure}.

            \item
            If $B = C[\semcka{e}]$ such that $e \leq f \in H$ and $C[\semcka{f}] \subseteq \closure{K}$, then by the induction hypothesis we have $L \cdot C[\semcka{f}] \subseteq \closurep{L \cdot K}$.
            Let $U \in L\cdot C[\semcka e]$; there are $V \in L$ and $W\in C[\semcka{e}]$ such that $U = V \cdot W$.
            If we pick the context $C' = V \cdot C$, then we have $U \in C'[\semcka{e}]$.
            Since $V \in L$, we also have $C'[\semcka{f}] \subseteq L \cdot C[\semcka{f}] \subseteq \closurep{L \cdot K}$.
            Therefore $U \in C'[\semcka e] \subseteq \closurep{L \cdot K}$, and hence $L \cdot B = L \cdot C[\semcka{e}] \subseteq \closurep{L \cdot K}$.
        \end{itemize}

        \item
        If $A = C[\semcka{e}]$ s.t.\ $e \leq f\in H$ and $C[\semcka{f}] \subseteq \closure{L}$, then by the induction hypothesis we have $C[\semcka{f}] \cdot \closure{K} \subseteq \closurep{L\cdot K}$.
        Let $U \in C[\semcka{e}] \cdot \closure{K}$; there are $V \in C[\semcka{e}]$ and $W \in \closure{K}$ such that $U = V \cdot W$.
        If we pick the context $C' = C \cdot W$, then we have $U \in C'[\semcka{e}]$.
        Since $W\in K$, we also have $C'[\semcka{f}] \subseteq C[\semcka{f}] \cdot \closure{K} \subseteq \closurep{L \cdot K}$.
        Therefore $U \in C'[\semcka{e}] \subseteq \closurep{L \cdot K}$, and hence $A \cdot \closure{K} = C[\semcka{e}] \cdot \closure{K} \subseteq \closurep{L \cdot K}$.
    \end{itemize}

    \item
    This equation can be proved using an argument similar to the above.

    \item
    The inclusion from left to right follows by an argument similar that in~\eqref{property:hypothesis-concat}.

    \smallskip
    For the other inclusion, it suffices to prove that ${({\closure L})}^* \subseteq \closurep {L^*}$.
    To this end, we first argue that for all $n \in \naturals$, it holds that ${({\closure L})}^n \subseteq \closurep {L^*}$, by induction on $n$.
    In the base, where $n = 0$, we have that ${({\closure L})}^0 = \{ 1 \} \subseteq \closurep {\{ 1 \}} \subseteq \closurep {L^*}$ by~\eqref{property:hypothesis-monotone}.
    For the inductive step, suppose the claim holds for $n$.
    We then calculate, using~\eqref{property:hypothesis-closure} and~\eqref{property:hypothesis-concat}, that
    \[
            {({\closure L})}^{n+1}
                    = {({\closure L})}^n \cdot {\closure L}
                    \subseteq \closurep {L^*} \cdot {\closure L}
                    \subseteq \closurep {L^* \cdot L}
                    \subseteq \closurep {L^*}
    \]

    Putting this together, we have that
    \[
            {({\closure L})}^*
                    = \bigcup_{n \in \naturals} {({{\closure L}})}^n
                    \subseteq \bigcup_{n \in \naturals} \closurep {L^*}
                    \subseteq \closurep {L^*}
    \]

    \item
    Since $\closure L\subseteq\closure K$, we have $C[\closure L]\subseteq C[\closure K]$,
    therefore by~\eqref{property:hypothesis-monotone} we obtain $\closure{C[\closure L]}\subseteq \closure{C[\closure K]}$. From~\eqref{property:hypothesis-concat} and~\eqref{property:hypothesis-parallel} we get that
    \[
        \closure{C[\closure L]}=\closure{C[L]},
    \]
    so we may conclude:
    \[\closure{C[L]}=\closure{C[\closure L]}\subseteq \closure{C[\closure K]}
    =\closure{C[K]}.\]

    \item
    We proceed by induction on the construction of $\closure{L}$, showing that
    \[
        \forall A \subseteq \closure{L}, A \subseteq \SP
    \]
    \begin{itemize}
        \item
        In the base, we have that $A = L$, whence $A \subseteq \SP$ by the premise.

        \item
        For the inductive step, we have $A = C[\semcka{e}]$ with $e \leq f \in H$, and $C[\semcka{f}] \subseteq \closure{L}$.
        Since $C \in \PCsp$ and $\semcka{e} \subseteq \SP$, $C[\semcka{e}] \subseteq \SP$ by \autoref{lemma:context-preserve-series-parallel}.
        \qedhere
    \end{itemize}
\end{enumerate}
\end{proof}

\soundness*
\begin{proof}
The proof proceeds by induction on the construction of $\equivcka^H$.

In the base, there are two cases to consider.
On the one hand, if $e \equivcka^H f$ because $e \equivcka f$, then $\semcka{e} = \semcka{f}$ by \autoref{theorem:bka-completeness-decidability}; hence, $\closure {\semcka{e}} = \closure {\semcka{f}}$.
On the other hand, if $e \leqqcka^H f$ because $e \leq f \in H$, then it suffices to prove that $\semcka{e} \subseteq \closure {\semcka{f}}$, by \autoref{lemma:composition-vs-closure}\eqref{property:hypothesis-closure}.
If we choose $C = *$, then we find $C[\semcka{f}]=\semcka{f}\subseteq\closure{\semcka{f}}$.
By definition of closure, we then have $C[\semcka{e}]=\semcka{e}\subseteq\closure{\semcka{f}}$.

For the inductive step, suppose that $e \equivcka^H f$ because $e = e_0 + e_1$ and $f = f_0 + f_1$ such that $e_i \equivcka^H f_i$ for $i \in 2$.
By induction, we then know that $\closure {\semcka{e_i}} = \closure {\semcka{f_i}}$ for $i \in 2$.
Using \autoref{lemma:composition-vs-closure}\eqref{property:hypothesis-union}, we then derive as follows:
\begin{align*}
\closure {\semcka{e}}
    &= \closurep {\semcka{e_0} \cup \semcka{e_1}} \\
    &= \closurep {\closure {\semcka{e_0}} \cup \closure {\semcka{e_1}}} \\
    &= \closurep {\closure {\semcka{f_0}} \cup \closure {\semcka{f_1}}} \\
    &= \closurep {\semcka{f_0} \cup \semcka{f_1}} \\
    &= \closure {\semcka{f}}
\end{align*}
The other inductive steps that arise congruence can be argued similarly.

It remains to validate the fixpoint axioms.
Here, we have that $e \leqqcka^H f$ because $e = g^* \cdot h + f$ with $g \cdot f + h \leqqcka^H f$.
By \autoref{lemma:composition-vs-closure}\eqref{property:hypothesis-closure}, it suffices to prove $\semcka{g^* \cdot h} \subseteq \closure {\semcka{f}}$.
By induction, we know that $\closure {\semcka{h}} \subseteq \closure {\semcka{f}}$ and hence $\semcka{h} \subseteq \closure {\semcka{f}}$ by \autoref{lemma:composition-vs-closure}\eqref{property:hypothesis-closure}.
By \autoref{lemma:composition-vs-closure}\eqref{property:hypothesis-concat} and induction, we also find
\[
    \semcka{g} \cdot \closure {\semcka{f}}
        \subseteq \closure {\semcka{g}} \cdot \closure {\semcka{f}}
        \subseteq \closure {\semcka{g \cdot f}}
        \subseteq \closure {\semcka{f}}
\]
Putting this together, we conclude that $\semcka{g^* \cdot h} = \semcka{g}^* \cdot \semcka{h} \subseteq \closure {\semcka{f}}$.
\end{proof}

\implicationlemma*
\begin{proof}
We treat the claims in the order given.
\begin{enumerate}[(i)]
    \item
    By induction on $e \equivcka^{H'} f$.
    In the base, we have two cases.
    \begin{itemize}
        \item
        If $e \equivcka^{H'} f$ because $e \equivcka f$, then $e \equivcka^{H} f$ immediately.

        \item
        If $e \leqqcka^{H'} f$ because $e \leq f \in H'$, then $e \leqqcka^{H} f$ by the premise.
    \end{itemize}
    For the inductive step, there are again two cases to consider.
    \begin{itemize}
        \item
        If $e \equivcka^{H'} f$ because of a congruence rule, then the proof is straightforward.
        For instance, if $e = e_0 + e_1$ and $f = f_0 + f_1$ such that $e_i \equivcka^{H'} f_i$ for $i \in 2$, then by induction we know that $e_i \equivcka^{H} f_i$ for $i \in 2$.
        We can then conclude that $e \equivcka^{H} f$.

        \item
        If $e \leqqcka^{H'} f$ because of a fixpoint rule, then we proceed as follows.
        First, if $e = g \cdot h^*$ and $g + h \cdot f \leqqcka^{H'} f$, then by induction $g + h \cdot f \leqqcka^{H} f$.
        We then conclude that $e = g \cdot h^* \leqqcka^{H} f$ as well.
        The proof for the other fixpoint rule is similar.
    \end{itemize}

    \item
    We show that if $A \subseteq \closure[H']{L}$, then $A \subseteq \closure{L}$, by induction on $A \subseteq \closure[H']{L}$.
    \begin{itemize}
        \item
        In the base, where $A = L$, the claim follows by definition of closure.

        \item
        Otherwise, if there exist $C \in \PCsp$ and $e \leq f \in H'$ such that $A = C[\semcka{e}]$ and $C[\semcka{f}] \subseteq \closure[H']{L}$, then $C[\semcka{f}] \subseteq \closure{L}$ by induction.
        By the premise that $H$ implies $H'$ we furthermore know that $e \leqqcka^H f$, and hence $\closure{\semcka{e}} \subseteq \closure{\semcka{f}}$ by soundness.
        By \autoref{lemma:composition-vs-closure}\eqref{property:hypothesis-context}, we derive
        \[
            \closure{C[\semcka{e}]} \subseteq \closure{C[\semcka{f}]} \subseteq \closure{L}
        \]
        and hence $C[\semcka{e}] \subseteq \closure{L}$ by \autoref{lemma:composition-vs-closure}\eqref{property:hypothesis-closure}.
    \end{itemize}

    \item
    The inclusion from left to right follows by~\eqref{property:implication-vs-closure} and \autoref{lemma:composition-vs-closure}\eqref{property:hypothesis-closure}.
    For the other inclusion, it suffices to show that $L \subseteq \closurepsmall{\closure[H']{L}}$ by \autoref{lemma:composition-vs-closure}\eqref{property:hypothesis-closure}; this is true, since $L \subseteq \closure[H']{L} \subseteq \closurepsmall{\closure[H']{L}}$ by definition of closure.
    \qedhere
\end{enumerate}
\end{proof}

\carryingover*
\begin{proof}
Let $r$ be the reduction from $H$ to $H'$, and let $e, f \in \termscka$.

For completeness, suppose that $\closure[H]{\semcka{e}} = \closure[H]{\semcka{f}}$.
We then know that
\[
    \closure[H']{\semcka{r(e)}} = \closure[H']{\semcka{r(f)}}
\]
Hence, by completeness of $H'$, we have that $r(e) \equivcka^{H'} r(f)$.
Because $H$ implies $H'$, it follows that $r(e) \equivcka^H r(f)$.
We can then conclude by deriving
\[
    e \equivcka^H r(e) \equivcka^H r(f) \equivcka^H f
\]

For decidability, first note that we have that if $\closure{\semcka{e}} = \closure{\semcka{f}}$, then also $\closure[H']{\semcka{r(e)}} = \closure[H']{\semcka{r(f)}}$.
Conversely, if $\closure[H']{\semcka{r(e)}} = \closure[H']{\semcka{r(f)}}$, then
\begin{align*}
\closure{\semcka{e}}
    &= \closure{\semcka{r(e)}} \tag{Soundness} \\
    &= \closurepsmall{\closure[H']{\semcka{r(e)}}} \tag{\autoref{lemma:implication}\eqref{property:implication-vs-double-closure}} \\
    &= \closurepsmall{\closure[H']{\semcka{r(f)}}} \tag{Premise} \\
    &= \closure{\semcka{r(f)}} \tag{\autoref{lemma:implication}\eqref{property:implication-vs-double-closure}} \\
    &= \closure{\semcka{f}} \tag{Soundness}
\end{align*}
Hence, we can decide $\closure{\semcka{e}} = \closure{\semcka{f}}$ by checking if $\closure[H']{\semcka{r(e)}} = \closure[H']{\semcka{r(f)}}$.
\end{proof}

\chains*
\begin{proof}
Let $r$ be the reduction from $H$ to $H'$, and let $r'$ be the reduction from $H'$ to $H''$.
We claim that $r' \circ r$ is a reduction from $H$ to $H''$.

To see that $H$ implies $H''$, suppose that $e \leq f \in H''$.
Since $H'$ implies $H''$, we obtain $e \leqqcka^{H'} f$.
Since $H$ implies $H'$, we find $e \leqqcka^H f$ by \autoref{lemma:implication}\eqref{property:implication-vs-equivalence}.

To see that $e \equivcka^H r'(r(e))$, first note $e \equivcka^H r(e)$.
Also, $r(e) \equivcka^{H'} r'(r(e))$, and since $H$ implies $H'$, we have $r(e) \equivcka^H r'(r(e))$.
The claim then follows.

Lastly, suppose that $e, f \in \termscka$ such that $\closure{\semcka{e}} = \closure{\semcka{f}}$.
We then know that $\closure[H']{\semcka{r(e)}} = \closure[H']{\semcka{r(f)}}$, thus $\closure[H'']{\semcka{r'(r(e))}} = \closure[H'']{\semcka{r'(r(e))}}$.
\end{proof}

\factorisation*
\begin{proof}
Let $r_i$ be the strong reduction from $H_i$ to $H'$; choose $r = r_1 \circ r_0$.
We claim that $r$ is a strong reduction from $H$ to $H'$.
\begin{enumerate}[(i)]
    \item
    Since $H$ implies $H_1$, which in turn implies $H'$, we know that $H$ implies $H'$ by an argument similar to the one in \autoref{lemma:chains}.

    \item
    Let $e \in \termscka$.
    We then know that $e \equivcka^{H_0} r_0(e) \equivcka^{H_1} r_1(r_0(e)) = r(e)$.
    Since $H$ implies $H_0$ and $H_1$, we can conclude that $e \equivcka^H r(e)$ by \autoref{lemma:implication}\eqref{property:implication-vs-equivalence}.

    \item
    Let $e \in \termscka$.
    We then derive that
    \begin{align*}
    \closure[H]{\semcka{e}}
        &= \closure[H_1]{(\closure[H_0]{\semcka{e}})} \tag{Factorisation} \\
        &= \closure[H_1]{(\closure[H']{\semcka{r_0(e)}})} \tag{Reduction} \\
        &= \closure[H_1]{\semcka{r_0(e)}} \tag{\autoref{lemma:implication}\eqref{property:implication-vs-double-closure}} \\
        &= \closure[H']{\semcka{r_1(r_0(e))}} \tag{Reduction} \\
        &= \closure[H']{\semcka{r(e)}} \tag*{\qedhere}
    \end{align*}
\end{enumerate}
\end{proof}

\contextsem*
\begin{proof}
We prove the claims in the order given.
\begin{enumerate}[(i)]
    \item
    The proof proceeds by induction on the construction of $C$.
    In the base, $C = *$, in which case $r(*) = \{ * \} \subseteq \PCsp$.

    \smallskip
    For the inductive step, there are three cases to consider.
    If $C = C' \cdot V$, then $r(C) = r(C') \cdot r(V)$.
    By induction, we know that $r(C') \subseteq \PCsp$; since $r(V) \subseteq \SP$, the claim then follows by \autoref{lemma:series-parallel-contexts-inductive}.
    The other cases can be treated similarly.

    \item
    The proof proceeds by induction on the construction of $C$.
    In the base, $C = *$, in which case $r(C) = \{ * \}$, and hence $r(C[L]) = r(L) = \bigcup_{D \in r(C)} D[r(L)]$.

    \smallskip
    For the inductive step, there are three cases to consider.
    If $C = C' \cdot V$, then $r(C) = r(C') \cdot r(V)$.
    We make the following observations:
    \begin{itemize}
        \item
        If $D' \in r(C')$ and $U \in D'[r(L)] \cdot r(V)$, then there exists a $D \in r(C)$ such that $U \in D[r(L)]$.
        To see this, note that $U = V \cdot W$ for $V \in D'[r(L)]$ and $W \in r(V)$.
        If we then choose $D = D' \cdot W \in r(C') \cdot r(V) = r(C)$, we find that $U \in D'[r(L)] \cdot W = D[r(L)]$.

        \item
        If $D \in r(C)$, then there exists a $D' \in r(C')$ such that $D[r(L)] \subseteq D'[r(L)] \cdot r(V)$.
        To see this, note that $D = D' \cdot W$ for $D' \in r(C')$ and $W \in r(V)$, and that $D[r(L)] = D'[r(L)] \cdot W \subseteq D'[r(L)] \cdot r(V)$.
    \end{itemize}

    Hence, we derive that
    \begin{align*}
    r(C[L])
        &= r(C'[L] \cdot V) \tag{\autoref{lemma:substitution-vs-composition}\eqref{property:substitution-vs-concat-right}} \\
        &= r(C'[L]) \cdot r(V) \tag{Def. $r$ on languages} \\
        &= \Bigl( \bigcup\nolimits_{D' \in r(C')} D'[r(L)] \Bigr) \cdot r(V) \tag{Induction} \\
        &= \bigcup_{D' \in r(C')} D'[r(L)] \cdot r(V) \tag{Distributivity} \\
        &= \bigcup_{D \in r(C)} D[r(L)] \tag{Observations above}
    \end{align*}
    The other cases can be derived similarly.

    \item
    The proof proceeds by induction on the construction of $e$.
    In the base, there are two cases to consider.
    First, if $e = 0$ or $e = 1$, then $r(\semcka{e}) = \semcka{e} = \semcka{r(e)}$.
    Otherwise, if $e = \ltr{a}$ for some $\ltr{a} \in \Sigma$, then $r(\semcka{e}) = r(\{ \ltr{a} \}) = r(\ltr{a}) = \semcka{r(\ltr{a})}$.

    For the inductive step, the proof is straightforward.
    For instance, if $e = e_0 + e_1$, then we can derive that
    \begin{align*}
    r(\semcka{e_0 + e_1})
        &= r(\semcka{e_0} \cup \semcka{e_1})
            \tag{Def. $\semcka{-}$} \\
        &= r(\semcka{e_0}) \cup r(\semcka{e_1})
            \tag{Def. $r$ on languages} \\
        &= \semcka{r(e_0)} \cup \semcka{r(e_1)}
            \tag{Induction} \\
        &= \semcka{r(e_0) + r(e_1)}
            \tag{Def. $\semcka{-}$} \\
        &= \semcka{r(e_0 + e_1)}
            \tag{Def. $r$ on expressions}
    \end{align*}
    The other cases can be shown similarly.
    \qedhere
\end{enumerate}
\end{proof}

\propreif*
\begin{proof}
As usual for such statements, we proceed by induction on the construction of $\closure L$ the following statement:
\[
    \forall A \subseteq\closure L,\ r(A)\subseteq\closure[H']{r(L)}.
\]
\begin{itemize}
    \item
    In the base, where $A=L$, we have $r(L)\subseteq\closure[H']{r(L)}$ by definition of closure.

    \item
    For the inductive case, assume $A=C[\semcka e]$, for some $C$ and $e \leq f\in H'$ such that $C[\semcka{f}] \subseteq \closure{L}$.
    The induction hypothesis is that $r\left(C[\semcka{f}]\right) \subseteq \closure[H']{r(L)}$.
    Since $e\leq f\in H$, by definition of a reification we have $r(e) \leqqcka^{H'} r(f)$, so by soundness $\closure[H']{\semcka{r(e)}} \subseteq \closure[H']{\semcka{r(f)}}$.
    By \autoref{lemma:composition-vs-closure}\eqref{property:hypothesis-context}, for any context $D \in \PCsp$ this entails $\closure[H']{D\left[\semcka{r(e)}\right]} \subseteq \closure[H']{D\left[\semcka{r(f)}\right]}$.
    We may conclude:
    \begin{align*}
        r(C[\semcka{e}])
        &= \bigcup_{D\in r(C)} D[r(\semcka e)]
            \tag{\autoref{lemma:reify-context-sem}\eqref{property:reification-vs-context}} \\
        &\subseteq \bigcup_{D\in r(C)} \closure[H']{D[\semcka{r(e)}]}
            \tag{Def. closure} \\
        &\subseteq \bigcup_{D\in r(C)} \closure[H']{D[\semcka{r(f)}]}
            \tag{Observation above} \\
        &\subseteq \closurep[H']{\bigcup_{D\in r(C)} D[\semcka{r(f)}]}
            \tag{\autoref{lemma:composition-vs-closure}\eqref{property:hypothesis-union}} \\
        &= \closure[H']{r(C[\semcka{f}])}
            \tag{\autoref{lemma:reify-context-sem}\eqref{property:reification-vs-context}} \\
        &\subseteq \closurep[H']{\closure[H']{r(L)}}
            \tag{Induction} \\
        &\subseteq \closure[H']{r(L)}
            \tag*{(\autoref{lemma:composition-vs-closure}\eqref{property:hypothesis-closure}) \qedhere}
    \end{align*}
\end{itemize}
\end{proof}

\exchclosure*
\begin{proof}
We prove both directions separately.
\begin{itemize}
    \item
    For the implication from left to right, it is more convenient to reason about languages instead of individual pomsets.
    We write $L \subsp K$ if for every $U \in L$ there exists $V \in K$ such that $U \subsp V$; note that this makes $\subsp$ a preorder on languages.
    Using this definition we may reformulate the statement as:
    \[
        \forall A \subseteq \closure[\hexch]{L},\ A \subsp L.
    \]
    It should come as no surprise that we perform an induction on $A \subseteq \closure[\hexch]{L}$.
    \begin{itemize}
    \item The base case, where $A = L$, is trivial, since $\subsp$ is reflexive.
    \item For the inductive step, we have:
        \begin{mathpar}
        A = C[\semcka{(e \parallel f) \cdot (g \parallel h)}]
        \and
        C[\semcka{(e \cdot g) \parallel (f \cdot h)}] \subseteq \closure L.
        \end{mathpar}
        Our inductive hypothesis is $C[\semcka {(e \cdot g) \parallel (f \cdot h)}] \subsp L$.
        Now, since
        \[
            \semcka{(e \parallel f) \cdot (g \parallel h)}
                    \subsp \semcka{ (e \cdot g) \parallel (f \cdot h)},
        \]
        by definition of $\subsp$, we get
        \[
            C[\semcka {(e \cdot g) \parallel (f \cdot h)}]
                \subsp C[\semcka {(e \cdot g) \parallel (f \cdot h)}]
                \subsp M.
        \]
        by \autoref{lemma:pomset-context-monotone}.
        Therefore we conclude by transitivity.
    \end{itemize}
\item
    For the other direction, we first prove the following claim: if $C \in \PCsp$ and $U, V \in \SP$ such that $U \subsp V$ and $ C[V]\in \closure[\hexch]L$, then $ C[U]\in \closure[\hexch]L$, by induction on the construction of $\subsp$.
    In the base, there are two cases.
    \begin{itemize}
        \item
        If $U \subsp V$ because $U = V$, we find $C[U] = C[V] \in \closure[\hexch]L$ immediately.

        \item
        If $U \subsp V$ because there exist $W_{00}, W_{01}, W_{10}, W_{11} \in \SP $ such that
        \begin{mathpar}
        U = (W_{00} \parallel W_{01}) \cdot (W_{10} \parallel W_{11})
        \and
        V = (W_{00} \cdot W_{10}) \parallel (W_{01} \cdot W_{11})
        \end{mathpar}
        then we can find for each $i, j \in 2$ a $g_{ij} \in \termscka$ such that $\semcka{g_{ij}} = \{ W_{ij} \}$.
        We choose $e = (g_{00} \parallel g_{01}) \cdot (g_{10} \parallel g_{11})$ and $f = (g_{00} \cdot g_{10}) \parallel (g_{01} \cdot g_{11})$ to find that $e \leq f \in \hexch$, $\semcka{e} = \{ U \}$ and $\semcka{f} = \{ V \}$.
        By definition of $\closure[\hexch]L$, since $C[\semcka{f}] \subseteq \closure[\hexch]{L}$ it follows that $ C[U]\in C[\semcka{e}] \subseteq \closure[\hexch]L$.
    \end{itemize}
    For the inductive step, there are four cases to consider.
    \begin{itemize}
        \item
        If $U \subsp V$ because $U = U' \cdot W$ and $V = V' \cdot W$ with $U' \subsp V'$, then choose $C' = C[* \cdot W] \in \PCsp$.
        Since $U' \subsp V'$ and $C'[V'] = C[V' \cdot W] = C[V] \in L$, we find that $ C[U]= C[U' \cdot W] = C'[U'] \in \closure[\hexch]L$ by induction.

        \item
        If $U \subsp V$ because $U = W \cdot U'$ and $V = W \cdot V'$ with $U' \subsp V'$, the proof proceeds as above.

        \item
        If $U \subsp V$ because $U = W \parallel U'$ and $V = W \parallel V'$ with $U' \subsp V'$, the proof proceeds as above.

        \item
        If $U \subsp V$ because there exists a $W \in \Pom$ and $U \subsp W$ and $W \subsp V$, then by induction we first find that $ C[W]\in \closure[\hexch]L$, and if we apply the induction hypothesis once more can conclude that $ C[U]\in \closure[\hexch]L$.
    \end{itemize}
    Thus, if $V \in L$ with $U \subsp V$, we can choose $C = *$ to find that $ C[V]= V \in \closure[\hexch]L$, and hence $ U= C[U] \in \closure[\hexch]L$.
    \qedhere
\end{itemize}
\end{proof}

\factoriseseq*
\begin{proof}
We may reformulate the claim as
\[
    U \subsp V \in B \subseteq \closure[H]{(\closure[\hexch]{L})} \Rightarrow U \in \closure[H]{(\closure[\hexch]{L})}
\]
We proceed by induction on $B \subseteq \closure{(\closure[\hexch]L)}$.
In the base, $B \subseteq \closure{(\closure[\hexch]L)}$ because $B = \closure[\hexch]{L}$.
By \autoref{lemma:exch-closure-vs-subsumption}, we find that $U \in \closure[\hexch]{L}$, and thus $U \in \closure[H]{(\closure[\hexch]{L})}$.

For the inductive step, we obtain $e \leq f \in H$ and $C \in \PCsp$ such that $B = C[\semcka{e}]$, and $C[\semcka{f}] \subseteq \closure[H]{(\closure[\hexch]{L})}$.
There are two cases to consider.
\begin{itemize}
    \item
    If $e = \ltr{a}$ for some $\ltr{a} \in \Sigma$, then $B = \{ C[\ltr{a}] \}$, and hence $U \sqsubseteq C[\ltr{a}]$.
    By \autoref{lemma:context-subsumption}, we find $C' \in \PCsp$ such that $C' \sqsubseteq C$ and  $C'[\ltr{a}] = U$.
    By induction, and the fact that $C'[\semcka{f}] \subsp C[\semcka{f}]$ by \autoref{lemma:pomset-context-monotone}, it follows that $C'[\semcka{f}] \subseteq \closure{(\closure[\hexch]L)}$.
    Since $e \leq f \in H$ we can conclude that $U \in C'[\semcka{e}] \subseteq \closurep{\closure[\hexch]L}$.

    \item
    If $e = 1$, then $U \sqsubseteq C[1]$ then $B = \{ C[1] \}$, and hence $U \sqsubseteq C[1]$.
    By \autoref{lemma:context-erasure}, we find $C' \in \PCsp$ such that $C' \sqsubseteq C$ and $C'[1] = U$.
    By induction, and the fact that $C'[\semcka{f}] \subsp C[\semcka{f}]$ by \autoref{lemma:pomset-context-monotone}, it follows that $C'[\semcka{f}] \subseteq \closure{(\closure[\hexch]L)}$.
    Since $e \leq f \in H$, we can conclude that $U \in C'[\semcka{e}] \subseteq \closure{(\closure[\hexch]L)}$.
    \qedhere
\end{itemize}
\end{proof}

The following auxiliary lemma will be useful to prove \autoref{lemma:lift-closure}.
It says that substituting a sequential pomset in a sequential pomset yields a sequential pomset; conversely, a substitution that yielded a sequential context after plugging in a non-empty pomset must have come from a sequential pomset and context.

\begin{restatable}{lemma}{contextvssequential}%
\label{lemma:context-vs-sequential}
Let $C \in \PC$ and $U, V \in \Pom$.
The following hold:
\begin{enumerate}[(i)]
    \item\label{property:context-vs-sequential-preserve}
    If $U \in \Sigma^*$ and $C \in \PCseq$, then $C[U] \in \Sigma^*$.

    \item\label{property:context-vs-sequential-extract}
    If $C[U] \in \Sigma^*$ and $U \neq 1$, then $U \in \Sigma^*$ and $C \in \PCseq$.
\end{enumerate}
\end{restatable}
\begin{proof}
Let $C = [\lp{c}]$ and $U = [\lp{u}]$.
We treat the claims in the order given.
\begin{enumerate}[(i)]
    \item
    Let $s, s' \in S_{\lp{c}[\lp{u}]}$.
    We have three cases to consider.
    \begin{itemize}
        \item
        If $s, s' \in S_\lp{c} - \{ s_* \}$, then $s \leq_\lp{c} s'$ or $s' \leq_\lp{c} s$ because $C$ is totally ordered; hence, $s \leq_{\lp{c}[\lp{u}]} s'$ or $s' \leq_{\lp{c}[\lp{u}]} s$.

        \item
        If $s \in S_\lp{c} - \{ s_* \}$ and $s' \in S_\lp{u}$, then note that since $S_\lp{c}$ is totally ordered we have that either $s \leq_\lp{c} s_*$ or $s_* \leq_\lp{c} s$; hence, we have that $s \leq_{\lp{c}[\lp{u}]} s'$ or $s' \leq_{\lp{c}[\lp{u}]} s$ by definition of $\leq_{\lp{c}[\lp{u}]}$.

        \item
        If $s, s' \in S_\lp{u}$, then either $s \leq_\lp{u} s'$ or $s' \leq_\lp{u} s$ because $U$ is totally ordered; hence $s \leq_{\lp{c}[\lp{u}]} s'$ or $s' \leq_{\lp{c}[\lp{u}]} s$.
        \qedhere
    \end{itemize}

    \item
    To see that $U \in \Sigma^*$, let $s, s' \in S_\lp{u} \subseteq S_{\lp{c}[\lp{u}]}$.
    Because $C[U] \in \Sigma^*$, we have $s \leq_{\lp{c}[\lp{u}]} s'$ or $s' \leq_{\lp{c}[\lp{u}]} s$, hence $s \leq_\lp{u} s'$ or $s' \leq_\lp{u} s$.

    \smallskip
    To see that $C \in \PCseq$, let $s, s' \in S_\lp{c}$; we have three cases to consider.
    \begin{itemize}
        \item
        First, if $s = s_* = s'$, then $s \leq_\lp{c} s'$ immediately.

        \item
        Second, if $s = s_*$ and $s' \in S_\lp{c} - \{ s_* \}$, then take $s'' \in S_\lp{u}$, which exists because $S_\lp{u}$ is nonempty.
        Since $s, s'' \in S_{\lp{c}[\lp{u}]}$ and $\lp{c}[\lp{u}]$ is totally ordered, we have either $s \leq_{\lp{c}[\lp{u}]} s''$ or $s'' \leq_{\lp{c}[\lp{u}]} s$.
        In the former case, we find that $s \leq_\lp{c} s_*$ by definition of $\leq_{\lp{c}[\lp{u}]}$; the latter case can be treated similarly.

        \item
        Lastly, if $s, s' \neq s_*$, then since $s, s' \in S_{\lp{c}[\lp{u}]}$ and $\lp{c}[\lp{u}]$ is totally ordered, we have that either $s \leq_{\lp{c}[\lp{u}]} s'$ or $s' \leq_{\lp{c}[\lp{u}]} s$.
        By definition of $\leq_{\lp{c}[\lp{u}]}$ we then find that either $s \leq_\lp{c} s'$ or $s' \leq_\lp{c} s$.
    \end{itemize}
\end{enumerate}
\end{proof}

\liftclosure*
\begin{proof}
We treat the claims in the order given.
\begin{enumerate}[(i)]
    \item
    The inclusion from right to left is straightforward: if $A \subseteq \seqclosure{L}$, then $A \subseteq \closure{L}$ as well.
    For the other inclusion, suppose that $A \subseteq \closure{L}$.
    We proceed by induction on the construction of $A \subseteq \closure{L}$, showing that $A \subseteq \seqclosure{L}$ and $A \subseteq \Sigma^*$.
    In the base, know that $A \subseteq L$, hence $A \subseteq \seqclosure{L}$ and $A \subseteq \Sigma^*$.

    \smallskip
    For the inductive step, we find $e \leq f \in H$ and $C \in \PCsp$ such that $A = C[\semcka{e}]$ and $C[\semcka{f}] \subseteq \closure L$.
    Since $H$ is grounded, we have $\semcka{f} = \{ X \}$ for some non-empty word $X$.
    Since $C[X] \in C[\semcka{f}] \subseteq \Sigma^*$ by induction, it follows that $C \in \PCseq$ by \autoref{lemma:context-vs-sequential}\eqref{property:context-vs-sequential-extract}.
    Also by induction, we know that $C[\semcka{f}] \subseteq \seqclosure L$; hence, $A = C[\semcka{e}] \subseteq \closure L$.
    Finally, since $e \in \termska$, we have that $\semka{e} \in \Sigma^*$, and hence $A = C[\semcka{e}] \subseteq \Sigma^*$ as well, by \autoref{lemma:context-vs-sequential}\eqref{property:context-vs-sequential-preserve}.

    \item
    The inclusion from right to left follows from \autoref{lemma:composition-vs-closure}\eqref{property:hypothesis-parallel}.

    \smallskip
    For the other inclusion, suppose $A \subseteq \closurep{L \parallel L'}$; it suffices to show that we can find $B, B' \subseteq \SP$ such that $A \subseteq B \parallel B'$ and $B \subseteq \closure L$ and $B' \subseteq \closure {L'}$.
    We proceed by induction on the construction of $\closurep{L \parallel L'}$.
    In the base, where $A \subseteq L \parallel L'$, we can choose $B = L$ and $B' = L'$ to satisfy the claim.

    \smallskip
    For the inductive step, $A \subseteq \closurep{L \parallel L'}$ because there exists a $C \in \PCsp$ and $e \leq f \in H$ such that $A = C[\semcka{e}]$ and $C[\semcka{f}] \subseteq \closurep{L \parallel L'}$.
    By induction, we find $B, B' \subseteq \SP$ such that $C[\semcka{f}] \subseteq B \parallel B'$ and $B \subseteq \closure{L}$ and $B' \subseteq \closure{L'}$.
    Since $e \leq f$ is grounded, we know that $\semcka{f} = \{ W \}$ for some non-empty word $W$; hence $C[W] = X \parallel X'$ with $X \in B$ and $X' \in B'$.
    By \autoref{lemma:context-prime-locus}, we find that either $C = C' \parallel X'$ such that $C'[W] = X$, or $C = C' \parallel X$ such that $C'[W] = X'$.
    In the former case, we can write $A = C[\semcka{e}] \subseteq C'[\semcka{e}] \parallel B'$.
    Since $C'[\semcka{e}] \subseteq \closure{L}$ by definition of closure, the claim then follows.
    The latter case can be treated similarly.
    \qedhere
\end{enumerate}
\end{proof}

\liftreduction*
\begin{proof}
Let $r$ be the sequential reduction from $H$ to $H'$.
We extend $r$ to a map $\termscka \to \termscka$ by acting homomorphically, i.e., $r(e \parallel f) = r(e) \parallel r(f)$.
We already know that $H$ implies $H'$; it is not hard to show that if $e \in \termscka$, then $e \equivcka^H r(e)$.

For the last requirement, the proof proceeds by induction on the number of occurrences of $\parallel$ in $e$.
In the base, where $\parallel$ does not occur in $e$, we have that $e \in \termska$.
We can then derive by \autoref{lemma:lift-closure} that
\[
    \closure{\semka{e}}
        = \seqclosure {\semka{e}}
        = \seqclosure[H']{\semcka{r(e)}}
        = \closure[H']{\semcka{r(e)}}
\]
For the inductive step, we have $e = e_0 \parallel e_1$.
We then derive:
\begin{align*}
\closure {\semcka{e}}
    &= \closurep{\semcka{e_0} \parallel \semcka{e_1}}
        \tag{Def. $\semcka{-}$} \\
    &= \closure{\semcka{e_0}} \parallel \closure{\semcka{e_1}}
        \tag{\autoref{lemma:lift-closure}} \\
    &= \closure[H']{\semcka{r(e_0)}} \parallel \closure[H']{\semcka{r(e_1)}}
        \tag{Induction} \\
    &= \closurep[H']{\semcka{r(e_0)} \parallel \semcka{r(e_1)}}
        \tag{\autoref{lemma:lift-closure}} \\
    &= \closure[H']{\semcka{r(e)}}
        \tag*{(Def. $\semcka{-}$) \qedhere}
\end{align*}
\end{proof}

\fi%
\end{document}